%% file: main.tex
\documentclass[runningheads]{llncs}

\usepackage{graphicx}
\usepackage{amsmath,amssymb,amsfonts}
\usepackage{algorithmic}
\usepackage{textcomp}
\usepackage{xcolor}
\usepackage{url}
\usepackage{booktabs}
\usepackage{array}
\usepackage{listings}
\usepackage{subcaption}

\usepackage{stmaryrd} 
\usepackage{mathtools} 
\usepackage{tikz} 
\usepackage{enumerate}
\usepackage{algorithm}
\usepackage{todonotes}
\usepackage{multirow}

\newcommand{\acignore}[1]{{}}

\newcommand{\code}[1]{\texttt{#1}}

\newcommand{\EnumType}{\mathit{Enum}}
\newcommand{\EffectStmt}{\mathcal{E}}
\newcommand{\SetExpr}{E}
\newcommand{\SetVar}{s}

\newif\ifextended
\extendedtrue

\begin{document}

\title{Verification of Configurable SRA Systems\ifextended: Extended Version\fi}

\author{Alessandro Cimatti\inst{1} \and
Alberto Griggio\inst{1} \and
Christian Lidström\inst{1} \and
\\ Gianluca Redondi\inst{1} \and
Dylan Trenti
\inst{1,2}
}
\authorrunning{Cimatti et al.}
%
\institute{Fondazione Bruno Kessler, Trento, Italy \\
\email{\{cimatti,griggio,clidstrom,gredondi,dtrenti\}@fbk.eu}
\and 
University of Udine, Italy 
\email{trenti.dylan@spes.uniud.it}
}
\maketitle

\begin{abstract}
Many digital systems are designed as collections of asynchronous processes
orchestrated by a domain-specific scheduler. 
The verification of such \emph{scheduler-restricted asynchronous} systems (SRA)
is challenging due to process-process and process-scheduler interactions.
In this paper, we tackle the problem of verifying \emph{configurable} SRA.
A configurable SRA describes an unbounded family of possible SRA,
each resulting from an instantiation satisfying given configuration constraints;
our goal is proving at once that every legal instantiation of a configurable 
SRA is correct. 
We propose a contract-based, deductive verification approach that combines
(i) compositional proof rules that abstract the scheduler to prove top-level invariant properties,
(ii) automatic summarizations of the methods invoked by the scheduler,
(iii) simplification with respect to the nature of the space of configurations.
The approach is grounded in (object-oriented) first order logic,
requires reasoning over quantified statements, and leverages the 
Dafny software verifier as a backend. 
An experimental evaluation on industrial case studies demonstrates
that the framework scales effectively and enables practical 
reasoning about complex parameterized behaviors.
\acignore{
}
\end{abstract}

\input{intro}
\input{relatedwork}

\input{logics}

\input{compositional-verification}

\input{implementation}

\input{experiments}

\input{conclusions}

\newpage
\bibliographystyle{splncs04}
\bibliography{references}


\clearpage

\ifextended
\input{appendix}

\input{appendix-experiments}

\fi

\end{document}

%% file: intro.tex
\section{Introduction}
\label{sec:introduction}


%
Many modern digital systems are designed as the combination of interacting processes
coordinated by a domain-specific scheduler to implement a cyclic execution pattern.
Within this design pattern, which we refer to as Scheduler-Restricted Asynchronous systems (SRA),
the scheduler orchestrates process execution in discrete phases, without preemptions,
so that the traditional interleaving model of concurrency is restricted to a
result in a structured duty cycle.
%
%
The SRA paradigm is adopted in various domains, including embedded control and systems on chip,
with the prominent example of SystemC~\cite{systemc-standard}, and has been the subject of significant research efforts in formal
verification~\cite{DBLP:journals/fac/CimattiGMRTT98,DBLP:conf/spin/CampanaCNR11,DBLP:conf/iscas/GrosseD03,systemc-tasche,DBLP:conf/models/HerberH14}.

%
In this work, we focus on \emph{configurable} SRA systems,
that are compact, parameterized descriptions of unbounded families of SRA.
A configurable SRA $\mathcal{S}$ defines a set of concrete SRA; 
for each configuration $\mathcal{C}$ satisfying given configuration constraints $\Gamma$, 
there is a corresponding SRA $\mathcal{S}[\mathcal{C}]$ resulting from a (legal)
instantiation of the parameters with the object and interconnections specified by $\mathcal{C}$.

%
Configurable design is widely adopted, for example in software product lines, firmware for heat pump control
or generic railways control procedures,
for its ability to factor out the commonalities of the possible products.
The W-development process~\cite{Wmodel08} is composed of a first phase,
where the general traits of the domain are designed and analyzed,
and several application-specific phases, where the characteristics of each product are chosen.

Our goal is to perform \emph{parameterized verification},
proving that a given generic, quantified specification $\varphi$ holds
for all possible instantiations, i.e. for all $\mathcal{C}$ such that $\mathcal{C} \models \Gamma$, 
$\mathcal{S}[\mathcal{C}] \models \varphi$. 
This contrasts with the simpler task of verifying each concrete SRA $\mathcal{S}[\mathcal{C}]$ separately,
where the processes and their interconnections are fixed. 
The advantage of parameterized verification is that the effort is carried out,
at the configurable SRA level, earlier in the development process, so that
the generic application can be certified once and for all. The verification
of each concrete SRA $\mathcal{S}[\mathcal{C}]$ is limited to checking that 
$\mathcal{C} \models \Gamma$.

%
We propose a deductive verification approach for the parameterized
verification of configurable SRA systems.
We adopt an object-oriented modeling language.
Classes represent process types
and are equipped with a state machine, whose transitions
are associated with guards and effects described as methods in an imperative-style language.
Distinguishing features include the ability to represent configurations,
to quantify over unbounded collections and domain-specific types,
and to define a generic scheduling policy and generic properties.

\acignore{The semantics is defined based on a given configuration,
and cleanly separates local instance execution from global coordination,
naturally capturing the cyclic scheduling patterns prevalent in embedded systems.}


%
\acignore{
Our approach to parameterized verification relies on the expressiveness
and flexibility of \emph{deductive} methods, in order to ensure 
the reproduction of compositional reasoning informally carried out by domain experts.
At the same time, the approach provides automation through
automatic generation of local contracts.}

Our main technical contributions are a contract-based, deductive approach that combines the following ingredients:
(i) in order to prove top-level invariant properties,
a set of compositional proof rules are used to abstract the details
of the scheduling policy and decompose global properties into per-class obligations;
(ii) an automated summarization of the methods invoked by the scheduler,
which supports an abstraction from the implementation details; and
(iii) model simplification with respect to the nature of the space of configurations,
in order to reduce the complexity of the verification conditions.

The verification framework was implemented on top of Dafny~\cite{dafny}, whose language
is very suitable to express the features of configurable SRA systems.
We automatically generate both the Dafny models as well as the 
contracts for the process implementations.
We then apply our compositional strategy 
to prove that a given specification $\varphi$ holds
for all the possible valid configurations
using the automatically-generated implementation contracts and 
a user-supplied (quantified) invariant,
by discharging a sequence of proof obligations using the Dafny prover.
\acignore{The encoding into Dafny of the constructs
of the high-level language was carefully designed in order to properly
activate the features of the underlying verification engines such as quantified reasoning.}

This paper generalizes our previous case-study works~\cite{cav2025,rssrail2025}
to a domain-independent framework for configurable SRA systems. In particular, 
we introduce a set-based formalization,
automatic generation of implementation contracts,
and a general compositional verification strategy.




We carried out an experimental evaluation over benchmarks from industrial railway control systems
and embedded system control. The analyzed models comprise tens of thousands of lines of code and
automatically generated 
contracts.
All of these were automatically proved by Dafny to hold on the corresponding implementations.
We also demonstrate the impact of the simplifications driven by the configuration constraints,
that were proved correct, and we established multiple system level safety properties
for \emph{all possible configurations}.
The results demonstrate that the approach is able to benefit from the expressiveness of
deductive verification without having to deal with manual annotations  sacrificing automation and scalability.

%
\paragraph{Outline.}
Section~\ref{subsec:related-work} reviews the related work. 
Section~\ref{sec:control-logics} introduces the syntax and semantics of configurable SRA systems.
Section~\ref{sec:compositional-verification} formalizes the problem and 
presents the compositional verification strategy.
Section~\ref{sec:implementation} describes the tool chain,
and Section~\ref{sec:experiments} presents the experimental evaluation.
Section~\ref{sec:conclusion} concludes with future work.
Additional material (example, contract-generation algorithm, and detailed
experimental data) is available in the appendix.

%% file: relatedwork.tex
\section{Related Work}
\label{subsec:related-work}
%
Software Product Lines (SPLs) study the design of families of products,
each defined by a selection of features~\cite{DBLP:books/daglib/p/BocklePL05}.
The SPL verification problem amounts to proving the correctness
of (every product in) the family~\cite{DeductiveSPL,DeductiveSPLIncremental,DBLP:conf/fmoods/GrulerLS08}.
Differently from our work, SPLs typically focus on models of software systems, 
not on processes orchestrated by a scheduler.
More importantly, feature models in SPLs are typically 
restricted to finite sets (and not unbounded families) of products.

%
Various works deal with the verification of
SystemC~\cite{systemc-tasche,DBLP:journals/tcad/CimattiNR13,DBLP:conf/spin/CampanaCNR11},
a well known design language for SRA systems.
Tasche  et al.~\cite{systemc-tasche} propose a deductive verification
framwork where VerCors~\cite{vercors} is used as a backend 
to verify SystemC programs.
The works in~\cite{DBLP:journals/tcad/CimattiNR13,DBLP:conf/spin/CampanaCNR11}
present fully algorithmic approaches based on explicit-state model checking 
and on the ESST extension with software model checking techniques~\cite{esst}.
The work in~\cite{DBLP:journals/fac/CimattiGMRTT98} discusses
the verification of railways interlockings as SRA systems
with a domain-specific scheduler.
All these works focus on verifying a single SRA,
composed by a finite number of statically arranged threads,
rather than an unbounded family of systems.

Also related are the deductive verification frameworks for
parameterized systems such as Ivy~\cite{ivy,MCMT,lambda,mypyvy}. 
Their models are directly expressed as logical formulae
with uninterpreted functions and quantifiers.
This is a crucial difference, as our framework allows the
users to write \emph{programs} in a control-logic language,
with native notions of execution cycle and synchronization primitives.
This design choice is important for industrial adoption, where
engineers are more comfortable writing code than quantified
first-order specifications.
We also automatically extract declarative
specifications in the form of contracts from the code, similarly to~\cite{rummer}. 
%

Cutoff results~\cite{cutoff} establish
that the verification of unbounded families can be reduced
to checking a finite number of instances.
These results typically are limited to fixed topologies with strong structural constraints.
In contrast, our framework can deal with
system topologies described by a set of configuration constraints,
that provide greater modeling flexibility.

Somewhat related is the IronFleet framework~\cite{ironfleet}, 
that leverages Dafny for mechanically verifying the safety and liveness of 
practical distributed system implementations.
IronFleet starts with a global abstract specification,
with a proof strategy based on refinement between state machines.
Our work deals with a collection of component programs,
and relies on compositional contract-based reasoning.
Finally, IronFleet's case studies are primarily distributed systems, 
whereas we focus on SRA systems.

Our previous papers~\cite{cav2025,rssrail2025} are direct predecessors of this work.
They focus on specific railway systems and toolchain instances,
while in this paper we present a general framework with
improved methodology.
Section~\ref{sec:experiments} quantifies the impact of these methodological changes.

%% file: logics.tex
\section{Configurable SRA Systems}
\label{sec:control-logics}
Informally, a configurable Scheduler-Restricted Asynchronous (SRA) system 
consists of a set of class
declarations together with a scheduler description 
that coordinates instances' executions. 
%
%
%
Classes serve as templates 
for generating processes, expressed as extended finite state machines (EFSMs).
Each class defines local behavior through transitions,
consisting of a guard (a boolean expression), an effect (a program statement),
and a specific \emph{scheduling phase} in which the transition can be executed.
The scheduler is also defined 
as an EFSM that coordinates the execution of class instances,
as depicted in Figure~\ref{fig:sched},
and whose locations correspond to scheduling phases.
The definitions of both the process classes and the scheduler 
are \emph{parameterized} in terms of (unbounded) sets of related objects.
An instantiation of a system is then given by a \emph{configuration},
specifying the concrete set of objects that are controlled by the scheduler 
and their mutual connections.

Scheduler transitions 
are divided into
self-loop transitions (from phase $p$ to itself), 
in which all process instances execute one of their transitions marked with $p$ (if any), 
and
phase change transitions, in which the internal state of the scheduler is updated, but no process transition is executed.
%
A sequence of scheduler transitions starting from 
a designated \emph{initial phase $l_0$}
and ending in 
a designated \emph{final phase $l_f$} then forms a \emph{scheduling cycle}.
A \emph{run} of the system consists of a sequence of scheduling cycles.

\subsection{Syntax}
\label{subsec:syntax}

\begin{figure}[t]
    \centering
    \begin{subfigure}[b]{0.35\textwidth}
        \centering
        \includegraphics[width=\linewidth,alt={Architectural view of SRA systems.}]{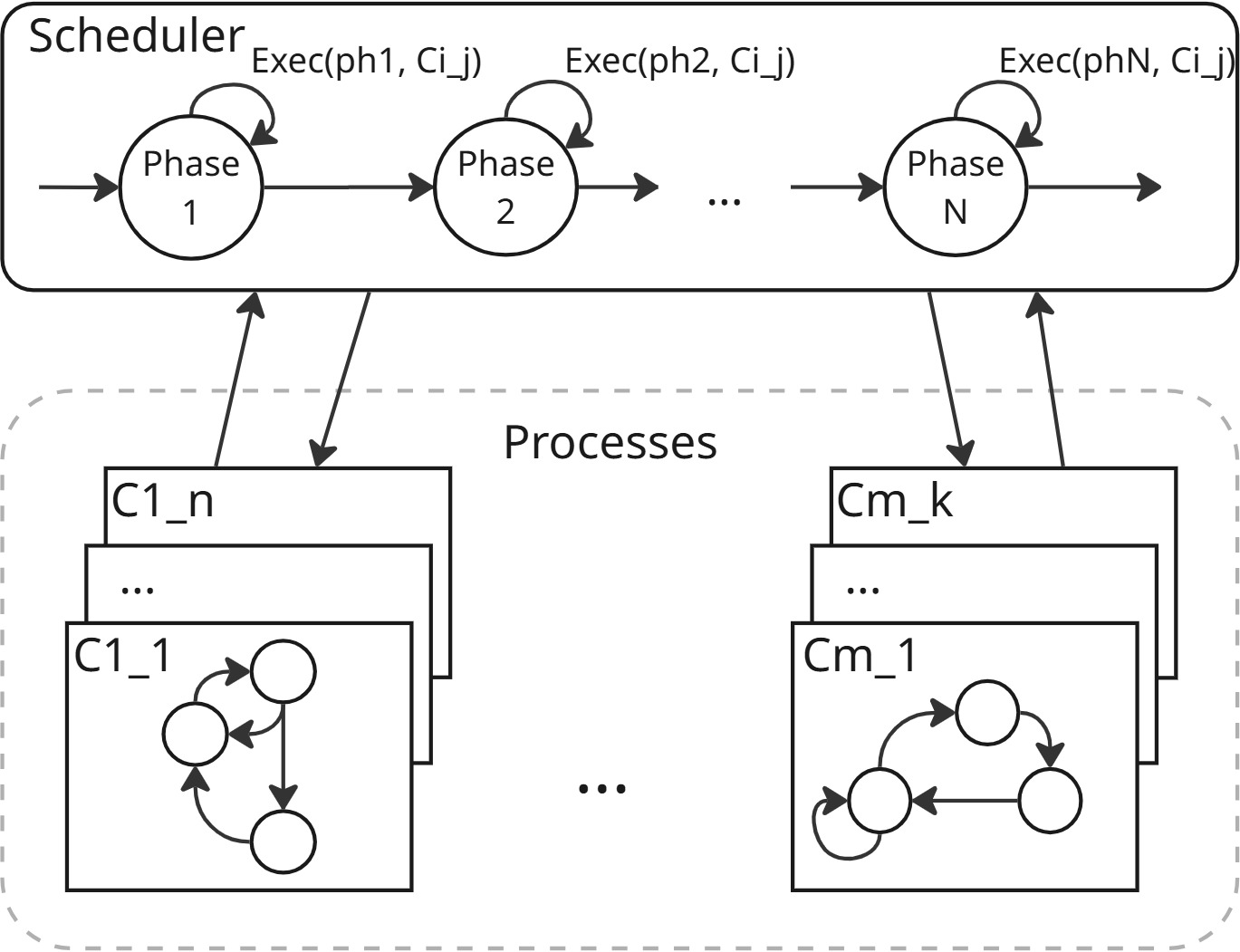}
    \end{subfigure}
    \hfill 
    \begin{subfigure}[b]{0.60\textwidth}
        \centering
        \includegraphics[width=\linewidth,alt={Dynamic view of SRA systems.}]{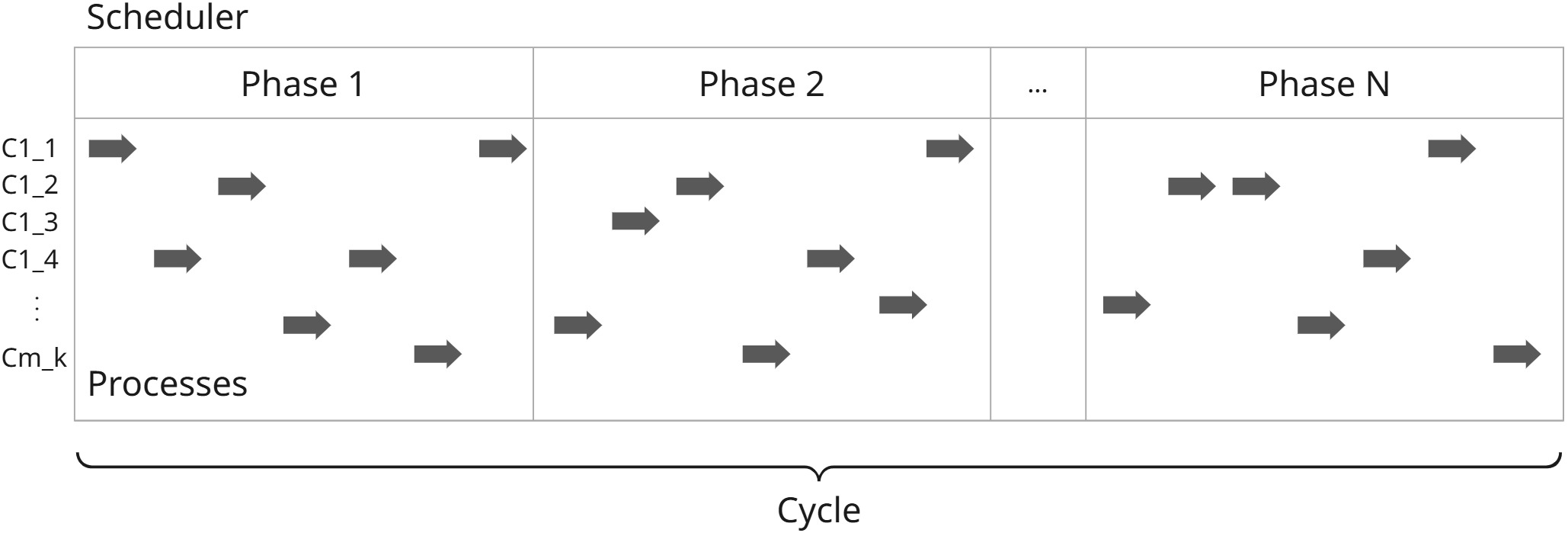}
    \end{subfigure}
    
    \caption{SRA: architectural view (left) and dynamic view (right).}
    \label{fig:sched}
\end{figure}


\begin{figure}[htb]
\small
\textbf{Types}\\[1ex]
\begin{minipage}[t]{0.5\linewidth}
$
\begin{array}{rcl}
\sigma &::=& \texttt{Int} \mid \texttt{Bool} \mid \EnumType \\
\tau &::=& \sigma \mid C \mid \texttt{Timer} \mid \texttt{Event} \mid \texttt{Set}\langle C\rangle
\end{array}
$
\end{minipage}\\[1em]
\begin{minipage}[t]{0.48\linewidth}
\textbf{Expressions}\\
\[
\begin{array}{rcl}
e &::=& c \mid x \mid x.f \mid e_1\pm e_2 \mid e_1*e_2\\
&\mid&  \textbf{if } b \textbf{ then } e_1 \textbf{ else } e_2 \mid |e| \\
b &::=& \texttt{true}\mid \texttt{false}\mid e_1=e_2\mid e_1<e_2 \\
&\mid& \lnot b\mid b_1\land b_2 \mid x\in e\mid e_1\subseteq e_2\mid e_1\;!!\;e_2\\
&\mid& e_1 \cup e_2 \mid (\forall x :: x\in e.\;b)\mid(\exists x :: x\in e.\;b)
\end{array}
\]
\end{minipage}\hfill
\begin{minipage}[t]{0.48\linewidth}
\textbf{Statements}\\
\[
\begin{array}{rcl}
S &::=& x := e; \mid x := *; \mid \textbf{if } b \textbf{ then } S \textbf{ else } S \\
  &\mid& \forall x:: x\in \SetExpr\ \{ x.f := e; \} \\
  &\mid& \texttt{assume } b; \mid \texttt{assert } b; \mid S;S \mid \varepsilon
\end{array}
\]
\end{minipage}
\caption{Types, expressions and statements.\label{fig:grammars}}
\end{figure}

\subsubsection{Type System}
\label{subsec:type-system}
The language provides basic types (\texttt{Int}, \texttt{Bool}), 
finite enumeration sorts, user-defined class types \(C\), 
and immutable set types \(\texttt{Set}\langle C\rangle\). 
Two additional special sorts model timing and events: \texttt{Timer} and \texttt{Event}; 
timers count the number of global execution cycles, while events are special 
boolean fields.
We also define a special enumeration type \texttt{PhaseEnum} that 
represents the scheduler's locations (phases), with distinguished 
initial and final values.

\subsubsection{Classes} 
\label{subsec:classes-structure}

A \emph{class} \(C\) is defined by a class declaration that 
specifies the structure and behavior of a finite state machine template. 
%
Each class declaration consists of variable declarations, 
parameter declarations and transition declarations. 
%
%
Variable declarations 
(\texttt{var} \(x : \sigma\)) introduce mutable scalar fields,
some of which can be tagged as \texttt{input}.
Moreover, each class has a special \texttt{location} variable of a finite enumeration type
(\texttt{LocationEnumC} for class type $C$) representing the locations of the EFSM.
Parameter declarations 
(\texttt{param} \(p : \sigma\)) introduce immutable constants, 
and set declarations (\texttt{set} \(s : \texttt{Set}\langle D \rangle\)) 
introduce immutable collections of objects, where \(D\) is any user-defined class type.
The \emph{class signature} $\Sigma_C$ of a class $C$ is the set of symbols
introduced by its declaration.
Each class also declares a finite set of \emph{transitions}.
A transition specifies a guard-effect pair that defines 
when and how state changes occur, and is given by 
a tuple \((l_{start}, G, l_{end}, \EffectStmt, p)\),
where \(l_{start}\) and \(l_{end}\) are class locations, 
\(G\) is a boolean expression (the guard), 
\(\EffectStmt\) is a statement (the effect), and \(p\) is a phase enumeration value 
indicating during which scheduler phase the transition is eligible for execution.
Every class contains three mandatory methods:
an \texttt{init} method for setting initial variable values,
an \texttt{exec} method for local execution (Section~\ref{subsec:local-semantics}),
and a \texttt{tick} method for updating timers.

\paragraph{Expression and statement language.}
Guards and effects in transitions are specified with expressions and statements in the language defined in Figure~\ref{fig:grammars}.
Basic expressions include constants (\(c\)), variables (\(x\)), field access (\(x.f\)),
and set cardinality (\(|e|\)). 
Boolean expressions extend standard logical operators with set membership (\(x \in e\)), 
subset relations (\(e_1 \subseteq e_2\)), disjointness assertions (\(e_1\;!!\;e_2\)),
and constrained quantification where \(\forall x :: x \in e .\; b\) and 
\(\exists x :: x \in e .\; b\) 
require \(e\) to have set type. We omit the description of the full type rules for brevity.
Quantified assignments \(\forall x :: x \in \SetExpr \; \{ x.f := e; \}\) are restricted 
to simple field assignments: \(\SetExpr\) must be of set type,
\(f\) must be a mutable scalar field of class \(C\), 
\(e\) is an expression that may reference \(x\) 
(i.e., no nested quantified assignments, method invocations, or other statements are allowed).

The language enforces several restrictions. 
The distinguished variable \texttt{location} 
cannot be assigned in transition effects: location changes occur implicitly when 
transitions execute. \texttt{Event} variables can only be assigned to 
\texttt{true} in effects; they are reset to \texttt{false} automatically when consumed by transitions.
Moreover, events are restricted syntactically: 
a guard is a conjunction of an event-free expression and a list of event variables.
Variables declared as \texttt{input} are externally controlled and cannot be assigned in effects. 
\subsubsection{Scheduler Signature and Transitions}
\label{subsec:scheduler-signature}
The scheduler signature $\Sigma_{\mathit{Sched}}$ extends the union of all class signatures with additional symbols
of the form $\texttt{All}_{C_i}$ of type $\texttt{Set}\langle C_i \rangle$ interpreted as the set of
instances of $C_i$ in the configuration, a variable
$\texttt{phase}$ of type $\texttt{PhaseEnum}$ for the scheduler's
current location, and a boolean variable $\texttt{executed}$
for each class instance.

Scheduler transitions define the global execution steps of the system. 
Each transition is specified by a guard, a start location $p$, 
and a target location $p'$. The guard is a $\Sigma_{\mathit{Sched}}$-expression, referencing only events, timers, 
and the \texttt{executed} flags of instances. 

Configuration constraints describe structural restrictions on admissible
configurations (e.g., disjointness relations or cardinality bounds).
They are formulas over $\Sigma_{\mathit{Sched}}$ that mention only immutable
symbols (set-valued fields and constant parameters).

\begin{definition}[Configurable SRA System]
A configurable Scheduler-Restricted Asynchronous system
\(\mathcal{S} = (\{C_1,\ldots,C_k\}, \mathit{Sched}, \Gamma)\) consists of a finite set of class
declarations \(\{C_1,\ldots,C_k\}\), a scheduler description
\(\mathit{Sched}\), and a set of configuration constraints \(\Gamma\)
over the immutable symbols.
\end{definition}

\subsection{Operational Semantics}
\label{subsec:semantics}

\begin{definition}[Configuration]
A \emph{configuration} \(\mathcal{C}\) for a finite set of classes 
\(\{C_1, \ldots, C_k\}\)
consists of:
(i) a finite \emph{universe} \(\mathcal{U}_i = \{o_{i,1}, \ldots, o_{i,n_i}\}\) 
of object instances for each class \(C_i\) (domain for sort \(C_i\)); and
(ii) an interpretation for all set-valued fields of \(C_i\) 
(interpretation for set predicates) and for all parameter fields of \(C_i\)
(constant values) for each instance \(o \in \mathcal{U}_i\).
\end{definition}
A configuration is the equivalent of a first-order structure, 
providing universes for each class and interpretations for set-valued fields 
and parameter fields (the non-mutable part of the system).

\subsubsection{Local Semantics}
\label{subsec:local-semantics}
The local semantics define how individual class instances behave
within a given configuration.
The local execution follows a \emph{small-step} semantics: 
each invocation of the \texttt{exec} method on a single instance 
constitutes one atomic step. 
The \texttt{exec} method takes a phase parameter that determines 
which category of transitions to consider, 
implementing the local execution through the following algorithm:
\begin{enumerate}
\item Filter all class transitions to include only those matching the 
specified phase;
\item Select the first transition whose guard is enabled in the current local state, 
according to the order in which transitions were declared in the class;
\item If an enabled transition is found: execute its effect statement sequence, 
update the control location to the target location, 
and reset to \texttt{false} every event variable that the 
occured in the transition's guard;
\item If no enabled transition exists: perform a \emph{stutter step} where the local state remains unchanged.
\end{enumerate}
\subsubsection{Global Execution Semantics}
\label{subsec:global-execution}

The global execution model coordinates class instances through the \emph{scheduler},
following a \emph{big-step} semantics.

\begin{definition}[Global State]
Given a configuration $\mathcal{C}$ with universe 
$\mathcal{U}$, 
a \emph{global state} $S$ is a function that assigns:
(i) for each instance $o \in \mathcal{U}$, 
    an interpretation for all scalar fields and set-valued fields of $o$; 
and (ii)
a location $S(\texttt{phase}) : \texttt{PhaseEnum}$, and 
a boolean value $S(\texttt{executed})$ for each instance.
\end{definition}

\noindent
%
The \emph{initial global state of the system} is established 
by initializing all instances and setting the scheduler to its initial location:
each instance $o$'s $\texttt{init}$ method    
sets its scalar fields to their declared initial values, 
timers to $\texttt{inactive}$, events to $\texttt{false}$, and $\texttt{executed}$ to $\texttt{false}$.

\paragraph{Global Execution.}
The effect of a scheduler transition from phase $p$ to $p'$ 
depends on its type:
\begin{itemize}
    \item \textbf{Self-loop transitions} ($p = p'$): The scheduler 
    invokes the \texttt{exec} method of each instance in an arbitrary order, 
    passing $p$ as the argument, and sets the $\texttt{executed}$ flag of each instance to $\texttt{true}$.   
    Note that local execution of an instance may read and modify both its own fields 
    and fields of other instances; therefore, the order in which instances are 
    executed can affect the resulting post-state.
    \item \textbf{Non-self-loop transitions} ($p \neq p'$): 
    The scheduler updates $\texttt{phase}$ to $p'$ and
    resets $\texttt{executed}$ to $\texttt{false}$ for all instances.
    If the target is the final location, 
    the scheduler additionally invokes each instance's \texttt{tick()} method to update timers.
\end{itemize}
\noindent
We denote a scheduler transition from global state $S$ to $S'$ as  
$S \xrightarrow{p \to p'} S'$. A visual representation of the global 
execution is given in Figure~\ref{fig:sched}, with the scheduler 
top coordinating the local execution of class instances within a cycle. 

\begin{definition}[Scheduling Cycle]
A \emph{scheduling cycle} is a sequence of global states 
$S_0, S_1, \ldots, S_n$ such that:
\begin{itemize}
    \item $S_0(\texttt{phase}) = l_0$ (the cycle starts at the initial location)
    \item For each $i < n$, we have 
    $S_i \xrightarrow{p_i \to p_{i+1}} S_{i+1}$
    for some scheduler transition
    \item $S_n(\texttt{phase}) = l_f$ (the cycle ends at the final location)
\end{itemize}
\end{definition}
After completing a cycle (reaching $l_f$), 
the scheduler transitions back to $l_0$ (setting $\texttt{phase}$ accordingly) to begin the next cycle;
in such transitions,
input variables of all instances
can be externally updated, while all other fields remain unchanged. 
\noindent
A \emph{run} of the system consists of a sequence of scheduling cycles,
starting with the initial global state.
An example of an SRA system is provided in Appendix~\ref{sec:appendix-example}.

%% file: compositional-verification.tex
\section{Compositional Verification of Configurable SRA}
\label{sec:compositional-verification}
We start by formalizing the verification problem 
for configurable SRA systems and then present a compositional verification strategy to solve it.
A single configurable SRA system describes an
unbounded family of possible SRA, each 
resulting from an instantiation satisfying its configuration constraints.
Our goal is proving at once that for every legal instantiation, 
at the end of every scheduling cycle, a given global property holds. 

\subsection{Verification Problem}
\emph{Global properties} $\varphi$ are safety assertions over the 
global state expressed as boolean $\Sigma_{\mathit{Sched}}$-expressions. 
They may refer to both immutable and mutable instance fields to capture invariants such as 
mutual exclusion or consistency requirements.


\begin{definition}[Parameterized Verification Problem]
Given a configurable SRA system $\mathcal{S}$ (with constraints $\Gamma$) and a global property $\varphi$, 
the \emph{parameterized safety verification problem} is to prove that, 
for all configurations $\mathcal{C}$ that satisfy all elements in $\Gamma$,
for each state $S$ in a run of the system such 
that $S(\texttt{phase}) = l_f$, the property holds:
$\mathcal{C}, S \models \varphi$.
\end{definition}

In this paper we focus on properties checked at the end of scheduling cycles
(i.e., when $\texttt{phase}=l_f$), but the same compositional obligations can be
adapted to properties required at other phases or over all reachable states.



\subsection{Compositional Verification}
\label{subsec:compositional-verification}
Our deductive verification strategy decomposes global reasoning
about the scheduler into a collection of local method contracts,
combined through first-order entailment checks. Rather than
modeling the scheduler explicitly, we reason compositionally: each
class method is verified in isolation, while global correctness is
obtained by proving that these local effects preserve a global
invariant $\mathit{Inv}$.

\subsubsection{Local Contracts}
\label{subsubsec:local-contracts}

A \emph{local contract} for a method $m$ of class $C$ is a Hoare
triple $\{\mathit{Pre}\}\; m \;\{\mathit{Post}\}$ where $\mathit{Pre}$
and $\mathit{Post}$ are $\Sigma_C$-expressions. 
Moreover, $\mathit{Post}$ may reference pre-state values 
via the \texttt{old} operator: for any expression~$e$,
$\texttt{old}(e)$ denotes the value of~$e$ immediately before the
method executes.
The contract is \emph{valid} if, for every configuration
$\mathcal{C}$ satisfying $\Gamma$ and every instance $o$ of $C$,
whenever $o$'s state satisfies $\mathit{Pre}$ before executing $m$,
it satisfies $\mathit{Post}$ afterwards.

In the following, we will focus on specific contracts that are used 
in our compositional argument. In particular, for each class $C$ and for 
each phase $p$ of the scheduler, we search for $\Sigma_C$-expressions $I_C$, $T_C(p)$, and $K_C$ such
that the following contracts are valid:
\[
\{\texttt{true}\}\ \texttt{init}()\ \{I_C\}
\qquad
\{\texttt{true}\}\ \texttt{exec}(p)\ \{T_C(p)\}
\qquad
\{\texttt{true}\}\ \texttt{tick}()\ \{K_C\}
\]
These contracts can be generated automatically from class definitions and verified independently
using standard program verification techniques.

\subsubsection{Global Entailment Checks}
\label{subsubsec:global-checks}

Using the local contracts above, we formulate a set of entailment
checks showing that, for every configuration satisfying~$\Gamma$, a
global invariant $\mathit{Inv}$ is preserved by every scheduler
transition and implies~$\varphi$ at the end of each cycle.
Sub-expressions wrapped
in $\texttt{old}(\cdot)$ are evaluated in the pre-state (before the
transition); all checks can be compiled to first-order validity queries
dischargeable by SMT solvers.
We write $\texttt{All}$ for the union of all $\texttt{All}_{C_i}$.

\paragraph{Self-loop execution preserves $\mathit{Inv}$.}
For each self-loop scheduler transition at phase~$p$, the scheduler invokes
\texttt{exec}($p$) on every instance in an \emph{arbitrary} order.
To show that $\mathit{Inv}$ is preserved by the whole interleaving,
we reduce the problem to a \emph{single-instance} check: it suffices
to show that executing any one instance preserves
$\mathit{Inv}$, then an induction on the number of executed instances
results in preservation for the full interleaving.

For a single-instance step we need to know what holds for an
instance~$c$ just before $c$~executes.
The global scheduler guard $G_{p\to p}$ holds at the start of the
self-loop, but it may be invalidated by earlier executions of other
instances.
We therefore introduce a \emph{local condition}
$g'_p$---a per-instance $\Sigma_C$-expression that approximates the
relevant information from the global guard and remains valid
throughout the interleaving.
Concretely, $g'_p$ must satisfy two side-conditions:
\begin{enumerate}
  \item[\textbf{(a)}] \emph{Establishment:}
    $\;\Gamma \models
        \big(\mathit{Inv} \land G_{p\to p}\big)
        \Rightarrow c.g'_p.$
  \item[\textbf{(b)}] \emph{Stability} (for every class $D$ and $d \neq c$ of type~$D$):
    $\;\Gamma \models
        \big(\texttt{old}(\mathit{Inv} \land c.g'_p)
          \land d.T_D(p)
          \land d.\texttt{executed}\big)
        \Rightarrow c.g'_p.$
\end{enumerate}
Given a valid $g'_p$, the main preservation check states that
executing a single instance~$c$ of class~$C$ under $g'_p$ preserves
$\mathit{Inv}$:
\[
\Gamma \models
\Big(
  \texttt{old}\big(\mathit{Inv}
    \;\land\; \texttt{phase} = p
    \;\land\; c.g'_p\big)
  \;\land\; c.T_C(p)
  \;\land\; c.\texttt{executed}
\Big)
\;\Rightarrow\; \mathit{Inv}
\]
\noindent
In practice, a common choice for $g'_p$ is
$\lnot c.\texttt{executed}$ (instance~$c$ has not yet executed in
phase~$p$), which is trivially stable under other instances'
executions.


We then check additional simpler preservation conditions for the 
initialization, for phase transitions, as well as the implication of $\varphi$ at the end of each cycle:

\paragraph{Initialization establishes $\mathit{Inv}$.}
%
\[
\Gamma \models
\Big(\forall c \in \texttt{All}.\;
  \big(c.I_C \;\land\; \lnot c.\texttt{executed}\big)
\;\land\; \texttt{phase} = l_0
\Big)
\;\Rightarrow\; \mathit{Inv}
\]

\paragraph{Phase transitions preserve $\mathit{Inv}$.}
For each scheduler edge $p \to p'$ with $p \neq p'$ and
guard~$g_{p \to p'}$:

\begin{itemize}
\item \emph{Non-final transition ($p' \neq l_f$).}
  The scheduler updates its control location and resets all
  \texttt{executed} flags:
  \begin{align*}
    \Gamma \models
    \big(\texttt{old}(\mathit{Inv} \;\land\; \texttt{phase} = p
      \;\land\; g_{p\to p'})
      \;\land\; \texttt{phase} = p'
      \;\land\; \\ \forall c \in \texttt{All}.\;
        \lnot c.\texttt{executed}\big)
    \;\Rightarrow\; \mathit{Inv}
  \end{align*}

\item \emph{Final transition ($p' = l_f$).}
  The scheduler additionally invokes
  \texttt{tick}() on each instance:
  \[
    \Gamma \models
    \left(
    \begin{aligned}
      &\texttt{old}(\mathit{Inv} \;\land\; \texttt{phase} = p
        \;\land\; g_{p\to l_f})
      \;\land\; \texttt{phase} = l_f \\
      &\land\; \forall c \in \texttt{All}.\;
        \big(\lnot c.\texttt{executed}
      \;\land\; c.K_C\big)
    \end{aligned}
    \right)
    \;\Rightarrow\; \mathit{Inv}
  \]

\item \emph{Reset ($l_f \to l_0$).}
  External inputs may change while non-input fields are preserved.
  Let $\mathrm{InputChange} \equiv
    \bigwedge_{\text{non-input } f}
    \forall c \in \texttt{All}.\; c.f = \texttt{old}(c.f)$.
  \[
    \Gamma \models
    \big(\texttt{old}(\mathit{Inv} \;\land\; \texttt{phase} = l_f)
      \;\land\; \mathrm{InputChange}
      \;\land\; \texttt{phase} = l_0\big)
    \;\Rightarrow\; \mathit{Inv}
  \]
\end{itemize}

\paragraph{$\mathit{Inv}$ implies the safety property.}
At the end of every cycle the invariant must imply the target property:
\[
  \Gamma \models
  \big(\mathit{Inv} \;\land\; \texttt{phase} = l_f\big)
  \;\Rightarrow\; \varphi
\]

\begin{theorem}[Compositional Soundness]
If all local contracts are valid and all global entailment checks
succeed for an invariant $\mathit{Inv}$, then $\varphi$ holds at the
end of every scheduling cycle in every run of~$\mathcal{S}$, for all
configurations satisfying~$\Gamma$.
\end{theorem}

\noindent
\emph{Proof sketch.}
The proof is by induction on scheduler transitions.
Initialization establishes $\mathit{Inv}$ by the initialization entailment check and valid
local \texttt{init} contracts; for preservation, we consider both phase-change scheduler
transitions and self-loop scheduler transitions, where the latter are further decomposed
into individual instance transitions (with an inner induction over the execution interleaving),
using the corresponding entailment checks together with valid local \texttt{exec} and \texttt{tick}
contracts; finally, at states with $\texttt{phase}=l_f$, $\varphi$ follows from
$\mathit{Inv}$ by the last entailment check.
The full proof is provided in Appendix~\ref{sec:proof-theorem}.

More generally, the problem of finding a suitable invariant can be 
reduced to the (undecidable) problem of safety verification of parameterized systems~\cite{parameterized-undecidability}.

%% file: implementation.tex
\section{Implementation in Dafny}
\label{sec:implementation}

The framework described in this paper has been implemented as a toolchain
built on top of the Dafny verification system~\cite{dafny}.
Dafny was chosen as it is a mature tool, which natively supports the combination 
of object-orientation, imperative statements, and quantified first-order logic, 
allowing a natural formalization of our models. 
The toolchain takes as input a description of a parameterized SRA system
written in a high-level control-logic language, and generates a SysML
description, executable C code, and a Dafny model for verification.
%

The implementation builds on an already existing toolchain.
In previous work~\cite{cav2025}, the framework focused on showing the correctness
of the generated executable C code, and as such the generated Dafny was made
to resemble this implementation. Components were connected through lists, 
and assignments were performed by iteration over the lists. 
The new implementation instead uses the set-theoretic formalization 
presented in Sec.~\ref{sec:control-logics}, and quantified assignments. 

Furthermore, safety properties were previously proven to be inductive for 
each  individual transition separately. Instead they are now proven over 
the \texttt{exec} method of each class, as described 
in Sec.~\ref{sec:compositional-verification}.


\subsection{Automatic Contract Generation}
\label{subsec:contract-generation}

Local contracts for methods can be generated automatically from the class definitions.
Thanks to the restrictions in our input language, this process is relatively straightforward. 
Due to lack of space, we cannot include the full details here,
but we report them in Appendix~\ref{sec:effect-transformation}.
Here, we only sketch the procedure for generating contracts for the \texttt{exec} method.
At a high level, 
this involves the following three steps:
(1) symbolic transformation of transition effects,
(2) encoding of individual transitions, and
(3) combination of the latter for the execution contract.
For each transition effect $\EffectStmt$, we compute a symbolic effect formula
$\text{Effect}_{\EffectStmt}$ that captures the state update performed by executing $\EffectStmt$.
This is done by forward symbolic execution, which computes a symbolic map $M_{\EffectStmt}$
associating each variable with an expression over pre-state values.
Assignments update the corresponding map entry; conditionals yield conditional
expressions; and quantified assignments generate lambda expressions for function
fields.
The effects of individual transitions are then combined together to obtain the 
contract for the \texttt{exec} method as a disjunction of the possible transitions, 
in which individual transition guards are augmented with constraints
enforcing the transition priority 
(a transition $t$ is to be executed only if there exists no other 
$t'$ declared before $t$ in the class definition whose guard evaluates to true).

\subsection{Optimization: Quantifier grounding}
\label{subsec:singleton-transformation}

If configuration constraints include constraints of the form $|\SetVar{}| = k$ for set-valued field $\SetVar{}$,
we can use this information to rewrite quantified expressions and statements into equivalent,
simpler forms.
Specifically, universal and existential quantifiers over $\SetVar{}$ can be expanded into conjunctions
or disjunctions over the $k$ elements of $\SetVar{}$. A similar encoding can be applied in case of 
bound constraints of the form $|\SetVar{}| \leq k$.
This transformation can be applied both at the level of code (rewriting quantified statements in method bodies)
and at the level of contracts (e.g., quantified postconditions),
resulting in quantifier-free verification conditions that are easier for SMT solvers to handle.
Furthermore, we can prove that these rewrites result in equivalent specifications, assuming
that the constraints hold. Thus, when proving satisfaction of the
quantifer grounded specifications, and the safety properties, 
the same holds for the original specifications.

We implement this optimization for constraints $k \leq 1$ in the Dafny encoding. 
We keep the original set-valued fields in the ghost state 
(i.e., state only accessible to specifications) 
and generate new  object-valued fields (which are possibly nullable for bounded constraints) 
as the grounded version. 
We generate both the original and grounded version
of every specification, 
over the respective fields.
Then, we add lemmas stating the equivalence of the two versions, with two types of assumptions:
(i) the sizes of the set-valued fields are according to the given constraints, and
(ii) the object of each set-valued field equals the respective object-valued field.
Such assumptions are added for every field-pair occuring in the specifications 
whose equivalence are being proven.
%
For example, for a quantifier grounded nullable 
field $\SetVar{}$ (i.e., now an object), we add also
the original field as ghost, $\SetVar{}_{\mathrm{ghost}}$. Then, for a grounded predicate $p$ referencing 
$\SetVar{}$ we generate also the original, non-grounded $p_{\mathrm{ghost}}$ referencing $\SetVar{}_{\mathrm{ghost}}$, 
and a lemma stating, where $|\SetVar{}_{\mathrm{ghost}}| \leq 1 \in \Gamma$:
$$\Gamma \land (|\SetVar{}_{\mathrm{ghost}}| = 0 \Rightarrow \SetVar{} = \mathrm{null})
 \land (|\SetVar{}_{\mathrm{ghost}}| = 1 \Rightarrow \{\SetVar{}\} = \SetVar{}_{\mathrm{ghost}}) \models p \iff p_{\mathrm{ghost}}.$$ 

%% file: experiments.tex
\section{Experimental Evaluation}
\label{sec:experiments}


\subsubsection*{Experimental Setup}

We evaluate the implementation of our framework on 
benchmarks
%
from previous work on SystemC verification~\cite{systemc-tasche}
and from industrial case studies obtained within a collaboration
with the Italian railway infrastructure operator RFI\footnote{Due
to the properietary nature of the applications, the case studies
cannot be publicly shared.}.
The verified applications were designed by the RFI signaling engineers
using the AIDA environment~\cite{DBLP:conf/isola/AmendolaBCCGSST20,DBLP:conf/sefm/CavadaCGS22},
while formal verification experts interacted with the toolchain described in this paper.
%
The cases used in this work have been presented in previous
application papers~\cite{cav2025,rssrail2025}, where verification
was performed with a preliminary version of the toolchain.
The present results were obtained with a consolidated and generalized version
of the framework, with the following differences.
First, collections of objects are represented as sets
(as well as lists); second, the toolchain implements a more general proof
strategy, that is not limited to the AIDA domain specific scheduling policy,
and that has been proved correct in the previous sections;
third, the quantifier grounding transformation provably results
in equivalent specifications (Section~\ref{subsec:singleton-transformation}).

We evaluate the toolchain on three different aspects: 
verification of \emph{local contracts} (i.e. all the methods
satisfy the automatically generated candidate contracts),
\emph{safety properties}, and \emph{quantifier grounding proofs}.
We report the number of assertions, and the performance measured in terms of 
time elapsed and of Dafny RC\footnote{The Dafny resource count
(RC) is a measure of the number of steps taken to complete the proof. RC is deterministic 
for a specific Dafny version.}.
The stated size of Dafny encodings includes the specifications.

We compare:
(i)~the current set-based formalization against the previous list-based formalization,
(ii)~the quantifier grounded encoding against the non-optimized one,
(iii)~the current proof strategy of proving safety properties over the 
\texttt{exec} method of each class, against the previous strategy
with proofs over individual transitions.
\acignore{The opaque block optimization employed in earlier work~\cite{cav2025} 
is not relevant in the set-based formalization and is not included here.}
%
%
%
All the experiments were performed on a cluster of identical nodes 
(AMD EPYC 7413 24-Core Processor, 96CPU, 1.0TB RAM), with
each verification task being assigned 8GB of RAM,
and using Dafny 4.11.0.

\subsubsection*{Experimental Results}
The verification results for local contracts of all systems are 
summarized in Table~\ref{tab:local-contracts}.
The \emph{total} time and resource count (RC) denotes 
verification of all methods against their contracts.
The \emph{execute} column is the verification of only the \texttt{exec}
method against its contract.
The \emph{guard} and \emph{effect} columns show only the verification of guards and
effects, respectively.
Note that the set formalization contain no guard methods.
%
The results for safety properties are summarized in Table~\ref{tab:props}.
The \emph{transition} columns are the results of verifying properties
over each of the transitions, whereas the \emph{execute} columns are
the results for verifying directly over the \texttt{exec} method
(which executes the individual transitions).
%
The detailed results are reported in Appendix~\ref{sec:appendix-experiments}.

\begin{table}[t]
    \caption{Verification statistics for local contracts}
    \centering
    \resizebox{0.85\textwidth}{!}{ 
    \begin{tabular}{c|c|c|rr|rr|rr|rr}
        \toprule
        Sys. &
        Version & \# Ass. & 
        Tot. time & Tot. RC & 
        Exec. time & Exec. RC & 
        Guard time & Guard RC & 
        Eff. time & Eff. RC \\
        \hline
    \multirow{1}{3em}{\centering RCS} &
    \multirow{1}{4em}{\centering Set}
     & 105 & 2s & 2.72M & 0s & 1.55M & 0s & 0.16M & 1s & 1.65M \\ 
        \hline
    \multirow{4}{3em}{\centering RPS} &
    \multirow{1}{4em}{\centering List} 
     & 994 & 4233s & 9436.56M & 350s & 724.29M & 3465s & 7034.48M & 396s & 1617.89M \\
    &\multirow{1}{4em}{\centering List+QG} 
     & 998 & 925s & 2529.53M & 403s & 671.75M & 60s & 188.41M & 442s & 1617.18M \\ 
    &\multirow{1}{4em}{\centering Set} 
     & 963 & 729s & 2154.06M & 345s & 675.18M & 7s & 21.37M & 356s & 1401.94M \\
    &\multirow{1}{4em}{\centering Set+QG}
     & 963 & 620s & 1851.96M & 234s & 388.14M & 5s & 13.7M & 362s & 1401.69M \\  
        \hline
    \multirow{4}{3em}{\centering SCL} &
    \multirow{1}{4em}{\centering List}
     & 408 & 1001s & 3814.77M & 97s & 346.38M & 1s & 2.4M & 895s & 3446.88M \\ 
    &\multirow{1}{4em}{\centering List+QG}
     & 411 & 924s & 3528.9M & 31s & 90.37M & 1s & 1.6M & 886s & 3418.88M \\ 
    &\multirow{1}{4em}{\centering Set}
     & 399 & 969s & 3763.57M & 82s & 317.41M & 1s & 1.91M & 880s & 3425.74M \\ 
    &\multirow{1}{4em}{\centering Set+QG}
     & 399 & 910s & 3518.73M & 22s & 80.44M & 1s & 1.5M & 880s & 3418.78M \\ 
        \bottomrule
    \end{tabular}
    }
    \label{tab:local-contracts}
\end{table}

\begin{table}[t]
    \caption{Verification statistics summary for safety properties}
    \centering
    \resizebox{
\textwidth}{!}{ 
    \begin{tabular}{c|c|rr|rr|rr|rr|rr|rr|rr|rr}
        \toprule
        \multirow{3}{2.2em}{\centering Sys.} & 
        \multirow{3}{4em}{\centering Prop.} & 
        \multicolumn{4}{c|}{\textbf{List}} & 
        \multicolumn{4}{c|}{\textbf{List+QG}} & 
        \multicolumn{4}{c|}{\textbf{Set}} & 
        \multicolumn{4}{c}{\textbf{Set+QG}} \\
        &&
        \multicolumn{2}{c}{Trans.} & \multicolumn{2}{c|}{Exec.} &
        \multicolumn{2}{c}{Trans.} & \multicolumn{2}{c|}{Exec.} &
        \multicolumn{2}{c}{Trans.} & \multicolumn{2}{c|}{Exec.} &
        \multicolumn{2}{c}{Trans.} & \multicolumn{2}{c}{Exec.} \\
        &&
        \multicolumn{1}{c}{Time} & \multicolumn{1}{c|}{RC} & 
        \multicolumn{1}{c}{Time} & \multicolumn{1}{c|}{RC} & 
        \multicolumn{1}{c}{Time} & \multicolumn{1}{c|}{RC} & 
        \multicolumn{1}{c}{Time} & \multicolumn{1}{c|}{RC} &
        \multicolumn{1}{c}{Time} & \multicolumn{1}{c|}{RC} & 
        \multicolumn{1}{c}{Time} & \multicolumn{1}{c|}{RC} & 
        \multicolumn{1}{c}{Time} & \multicolumn{1}{c|}{RC} & 
        \multicolumn{1}{c}{Time} & \multicolumn{1}{c}{RC} \\
        \hline
   \multirow{2}{2.2em}{RPS} & 
     Sum  &  16708s  &  42544M  &  755s  &  1899M  &  6195s  &  15508M  &  514s  &  1195M  &  12888s  &  33769M  &  797s  &  1908M  &  4096s  &  10064M  &  497s  &  1187M  \\
   & Max  &  2224s  &  5580M  &  116s  &  277M  &  370s  &  863M  &  72s  &  188M  &  1851s  &  4866M  &  92s  &  226M  &  238s  &  555M  &  59s  &  134M  \\
     \hline
     \multirow{1}{2.2em}{SCL} &
       1  &  13s  &  42M  & 11s  &  32M  &  12s  &  39M  & 13s  &  38M  &  12s  &  41M  & 12s  &  36M  &  11s  &  38M  & 13s  &  35M   \\
        \bottomrule
    \end{tabular}
    }
    \label{tab:props}
\end{table}

\subsubsection*{Generalized Robot Controller System}

The Robot Controller System (RCS) is a SystemC 
case study proposed in~\cite{systemc-tasche}, where a controller 
decides the movement based on the inputs of a \emph{fixed number of sensors}.
The property states that the controller never decides to move in a 
direction in which an obstacle is perceived. 
We generalize the RCS to include an unbounded number of sensors. 
The Dafny encoding is made up of about 800 LoC.
Verification of the safety property, which was not inductive and so was 
manually extended to form the invariant, took 18s and used 63.26M RC.
The full details of the example can be found in Appendix~\ref{sec:appendix-example}.

\subsubsection*{Railways Protection System (RPS)}
\label{sec:exp:pi}

The RPS~\cite{cav2025} is an industrial safety-critical system, developed by RFI, that ensures that workers
can safely access railway lines for maintenance.
We verify the control logic of the RPS, which is responsible for authorizing workers' requests to access
the line based on the state of the railway network.
The verification statistics are given 
for four different sets of generated codebases,
each between 11k and 14k LoC (with the list-based ones on the higher end).

The safety properties express the fact that authorization cannot be given 
in the presence of unsafe inputs.
We succesfully verify the same 21 safety properties 
as in~\cite{cav2025}, in addition to four new ones 
obtained from the railway engineers designing the system.
All but four of the properties were already inductive over the 
state machine transitions. 
The non-inductive properties were manually extended to form the invariant. 
The summary gives total time and resource count for all properties, as well as
the highest for a single property.

In the quantifier grounded version, verifying the 516 lemmas stating the equivalence
between quantifier grounded and non-optimized versions of specifications 
took 60s (with a total of 179M RC).

\subsubsection{Signal Control Logic (SCL)}
\label{sec:exp:segnale}

The SCL benchmark concerns the control logic for a railway signal system.
The SCL was originally described in~\cite{rssrail2025} to demonstrate the use of model
checking techniques implemented in nuXmv~\cite{nuXmv} for the generation of invariants
and counterexamples on manually generated finite abstractions of the SCL.
As such, no direct comparison with the techniques proposed here is possible.
We prove that at the end of each control cycle the signal aspects (colors) are set correctly.

Each generated code base is about 6k LoC.
%
In the quantifier grounded version, verification of the 210 equivalence lemmas
took 4s (with a total of 8M RC).






\subsubsection*{Discussion}

From the experimental results we can draw the following conclusions.
First, the approach is applicable in practice to prove safety properties
for all system configurations.
Second, the approach scales to realistic case studies.
The automated summarization of methods into candidate contracts
is effective in reducing the manual effort required for verification, and
requires reasonable computational effort.
%
Third, the set-theoretical formalization is superior to the list-based one in terms
of verification performance, hiding unnecessary details.
Similarly, the quantifier grounding optimization significantly reduces verification time 
and resource consumption in most cases, although the effect is more dramatic for
less optimized verification tasks (e.g., lists without grounding).
%
Finally, the verification of safety properties over the \texttt{exec} method
is significantly faster than over each transition. 
We conjecture that this is due to a better usage of the underlying
SMT solver in the \texttt{exec} version, which can reuse information across
multiple branches instead of restarting for each transition.
%
In general however, 
the careful control of the Dafny proof engine will require further investigation.

%% file: conclusions.tex
\section{Conclusion}
\label{sec:conclusion}

We presented a unified framework for the verification 
of configurable systems expressed as a collection of process types
orchestrated by a domain-specific scheduler. 
Each process type is described as an extended finite-state machine,
with actions attached to transitions expressed as imperative code.
First-order quantified expressions and statements
enable the direct logic-based representation of the space of possible configurations,
with variable number of processes and complex interconnections.
The verification approach is compositional and contract-based:
proving that a global safety property holds at the end of every scheduling cycle
is reduced to checking local contracts on each process type. The approach is implemented
on top of Dafny, and relies on the automatic generation of contracts
and annotations in order to reduce the manual effort required for verification.
The experimental evaluation, carried out on several real-world case studies,
demonstrated the generality, the effectiveness and the scalability of the approach.

%
There are several directions for future work.
First, we intend to further improve automation.
Currently, the decomposition of the global verification,
the extraction of contracts from methods implementations,
the application of configuration-specific optimizations 
and the compilation into Dafny all require no user intervention.
However, in general the user may have to provide
additional inductive invariants for the properties of interest.
While in our case studies the invariants were often straightforward extensions 
of the desired properties, this remains a manual step.
A promising direction is to leverage parameterized model checking techniques 
for automatic invariant discovery, e.g. \cite{lambda,duoai}, 
possibly extended with abstraction strategies similar to \cite{environment-abstraction,eagerabs}.
Our preliminary experiments in this direction show encouraging results,
and we plan to develop a fully automated pipeline in future work.
On a different direction, we intend to extend our approach to consider 
more general classes of properties, such as safety properties spanning
multiple scheduling cycles and liveness properties.
Finally, an important direction for future work will concern the application 
of formal or semi-formal techniques (e.g. equivalence checking, translation validation, or co-simulation)
to establish the functional equivalence between the models used for verification 
and the actual (automatically generated) C implementations deployed in production.


%% file: appendix.tex
\appendix

\section{Extended Example: Robot Controller System}
\label{sec:appendix-example}

We present an example of a configurable SRA system modeling a robot controller with multiple sensors,
inspired by a SystemC case study \cite{systemc-tasche}.
The system consists of a controller class and a sensor class, where the controller reads sensor states and decides on movement direction.
The scheduler has four phases: \texttt{Sense} (the intial phase), \texttt{Act}, \texttt{Reset}, and \texttt{End}.
\subsection{Class Definitions}
We start by defining the classes for the controller and sensors, along with their local transitions.
\textbf{Sensor Class:}
\begin{verbatim}
class Sensor {
  var location : SensLoc        // {Ready, Go, NoGo}
  var executed : Bool
  input obstacle : Bool         // external input
  event processed : Event       // set by controller, treated as a boolean
  
  transition senseGo = (Ready, !obstacle, Go, { }, Sense)
  transition senseNoGo = (Ready, obstacle, NoGo, { }, Sense)
  transition resetGo = (Go, processed, Ready, { }, Act)
  transition resetNoGo = (NoGo, processed, Ready, { }, Act)
}
\end{verbatim}

The \texttt{Sensor} class detects obstacles and transitions between control states:
\texttt{Ready} (waiting to sense), \texttt{Go} (no obstacle), or \texttt{NoGo} (obstacle detected).
It has an event \texttt{processed} set by the controller to acknowledge the sensor reading.
The initialize method (not shown) sets the location to \texttt{Ready}.

\textbf{Controller Class:}
\begin{verbatim}
class Controller {
  var location : CtrlLoc        // {Idle, Moving}
  var executed : Bool
  var direction : Direction     // {Forward, Stop, Left, Right}
  set leftSensors : Set<Sensor>
  set rightSensors : Set<Sensor>
  set allSensors : Set<Sensor> //{ leftSensors + rightSensors }
  
  transition actRight = (Idle, 
    (exists s in leftSensors : s.location == NoGo) && 
    (forall s in rightSensors : s.location == Go), 
    Moving, { direction := Right;
      forall s in allSensors { s.processed := true; } }, Act)
  transition actLeft = (Idle, 
    (exists s in rightSensors : s.location == NoGo) && 
    (forall s in leftSensors : s.location == Go), 
    Moving, { direction := Left;
      forall s in allSensors { s.processed := true; } }, Act)
  // similar transition for actForward and actStop, omitted for brevity
  transition reset = (Moving, true, Idle, { }, Reset)
  
}
\end{verbatim}

The \texttt{Controller} class owns sets of sensors and reads their state
via quantified expressions. It decides on movement direction based on sensor readings
and signals sensors when it has processed their state.
The initialize method (not shown) sets the location to \texttt{Idle} and direction to \texttt{Stop}.

\paragraph{Scheduler.}
The scheduler in this example has a set of locations
$\texttt{PhaseEnum} = \{\texttt{Sense}, \texttt{Act}, \texttt{Reset}, \texttt{End}\}$,
with initial phase $\texttt{Sense}$ and final phase $\texttt{End}$.
The scheduler transitions are defined as follows:
\begin{enumerate}
\item \emph{Self-loop on \texttt{Sense}} $(p = p' = \texttt{Sense})$, with guard
$G_1 \equiv \forall x \in \texttt{All}.\; \lnot x.\texttt{executed}$
This transition triggers the local execution of all instances in the \texttt{Sense} phase. 
Since at the start of each cycle $\texttt{executed} = \texttt{false}$ for all instances, 
the guard is satisfied.
After execution, each instance's $\texttt{executed}$ flag is set to $\texttt{true}$, 
disabling the guard.

\item \emph{Phase change from \texttt{Sense} to \texttt{Act}} $(p = \texttt{Sense},\; p' = \texttt{Act})$, with guard
 $ \lnot G_1$. This transition fires once all instances have executed in the \texttt{Sense} phase.

\item \emph{Self-loop on \texttt{Act}} $(p = p' = \texttt{Act})$, with guard:
$G_2 \equiv \forall x \in \texttt{All}.\; \big(\lnot x.\texttt{executed}\big) 
\;\lor\; (\exists s \in \texttt{All}_{\texttt{Sensor}}.\; s.\texttt{processed})$
This transition repeatedly invokes local execution in the \texttt{Act} phase. 
The guard ensures that the self-loop continues as long as either some instance has not yet executed in this phase 
or there remain sensors with a pending \texttt{processed} event, 
forcing the system to iterate until all events are consumed.

\item \emph{Phase change from \texttt{Act} to \texttt{Reset}} $(p = \texttt{Act},\; p' = \texttt{Reset})$, with guard $\lnot G_2$.
This transition fires once all instances have executed in the \texttt{Act} phase and no pending events remain.

\item \emph{Self-loop on \texttt{Reset}} $(p = p' = \texttt{Reset})$, with guard
$G_3 \equiv \forall x \in \texttt{All}.\; \lnot x.\texttt{executed}$
This transition triggers local execution for the \texttt{Reset} phase.

\item \emph{Phase change from \texttt{Reset} to \texttt{End}} $(p = \texttt{Reset},\; p' = \texttt{End})$, with guard $\lnot G_3$

\end{enumerate}
A scheduling cycle thus consists of the transition sequence
$S_0 \xrightarrow{\texttt{Sense} \to \texttt{Sense}} S_1 
\xrightarrow{\texttt{Sense} \to \texttt{Act}} S_2 
\xrightarrow{\texttt{Act} \to \texttt{Act}} \cdots 
\xrightarrow{\texttt{Act} \to \texttt{Reset}} S_j 
\xrightarrow{\texttt{Reset} \to \texttt{Reset}} S_{j+1} 
\xrightarrow{\texttt{Reset} \to \texttt{End}} S_n$,
where the self-loop on \texttt{Act} may iterate multiple times 
until all events are processed.

\paragraph{Configuration constraints.}
The system's structural constraints are given by the following configuration constraints:
\begin{align*}
\Gamma_1 &\equiv \forall c :: c \text{ in } \code{All}_{\code{Controller}}.\; c.\code{leftSensors} \mathbin{!!} c.\code{rightSensors} \\
\Gamma_2 &\equiv \forall c_1 :: c_1 \text{ in } \code{All}_{\code{Controller}}.\; \forall c_2 :: c_2 \text{ in } \code{All}_{\code{Controller}}.\; c_1 \neq c_2 \Rightarrow \\
& \quad (c_1.\code{leftSensors} \cup c_1.\code{rightSensors}) \mathbin{!!} (c_2.\code{leftSensors} \cup c_2.\code{rightSensors})\\
\Gamma_3 &\equiv \forall c :: c \text{ in } \code{All}_{\code{Controller}}.\; |c.\code{leftSensors}| \ge 1 \\
\Gamma_4 &\equiv \forall c :: c \text{ in } \code{All}_{\code{Controller}}.\; |c.\code{rightSensors}| \ge 1 \\
\Gamma_5 &\equiv \forall c :: c \text{ in } \code{All}_{\code{Controller}}.\; c.\code{leftSensors} \cup c.\code{rightSensors} = \code{allSensors} 
\end{align*}
We denote with $\Gamma$ the set $\{\Gamma_1, \Gamma_2, \Gamma_3, \Gamma_4, \Gamma_5\}$.
These constraints ensure every controller has at least one sensor on each side, 
that left and right sensor sets are disjoint within each controller, 
and that no sensor is shared between different controllers.
\paragraph{Global Safety Property.}
The system must satisfy the following safety property:
\[
\mathit{Prop} \equiv \forall c :: c \text{ in } \code{All}_{\code{Controller}}.\; ((\exists s :: s \text{ in } c.\code{leftSensors}.\; s.\code{obstacle}) \Rightarrow c.\code{direction} \neq \code{Left})
\]
This property states that if any left sensor has detected an obstacle, the controller is not moving left.
In our semantics, this property must hold in all reachable states at the end of each scheduler cycle.
In fact, the property can also be violated during intermediate phases within cycles.

\subsection{Configuration Example}
A simple configuration \(\mathcal{C}\) for the example consists of:
\begin{itemize}
\item One controller instance: \(\mathcal{U}_{\text{Controller}} = \{c_1\}\)
\item Three sensor instances: \(\mathcal{U}_{\text{Sensor}} = \{s_L, s_{R1}, s_{R2}\}\)
\item Set field interpretations: \(c_1.\text{leftSensors} = \{s_L\}\), \(c_1.\text{rightSensors} = \{s_{R1}, s_{R2}\}\), 
and \(c_1.\text{allSensors} = \{s_L, s_{R1}, s_{R2}\}\)
\end{itemize}
This configuration clearly satisfies the configuration constraints \(\Gamma\).

\subsection{Execution Scenario}
We illustrate a sample execution scenario over the above configuration \(\mathcal{C}\).
Initially, the controller \(c_1\) is in the \texttt{Idle} state, 
its direction is set to \texttt{Stop}, and \texttt{executed} is \texttt{false}. 
All sensors, namely \(s_L\), \(s_{R1}\), and \(s_{R2}\), are in the \texttt{Ready} state, 
each with \texttt{executed} set to \texttt{false}, and their \texttt{processed} event is false.
The initial scheduler phase is \texttt{Sense}.

For the external inputs, suppose that the left sensor detects an obstacle, i.e., 
\(s_L.\text{obstacle} = \text{true}\), while both right sensors detect no obstacle, 
so \(s_{R1}.\text{obstacle} = \text{false}\) and \(s_{R2}.\text{obstacle} = \text{false}\).

The execution proceeds in scheduler transitions. 
The first transition is the self-transition of the \texttt{Sense} phase, 
which execute the local transitions of all instances in the \texttt{Sense} phase.
In particular, we have:
\begin{itemize}
\item \emph{Sensors:} 
  \begin{itemize}
  \item $s_L$ detects an obstacle ($s_L.\text{obstacle} = \text{true}$), 
  so it executes \texttt{senseNoGo} transition: \texttt{Ready} $\to$ \texttt{NoGo}
  \item $s_{R1}$ detects no obstacle ($s_{R1}.\text{obstacle} = \text{false}$), 
  so it executes \texttt{senseGo} transition: \texttt{Ready} $\to$ \texttt{Go}
  \item $s_{R2}$ detects no obstacle ($s_{R2}.\text{obstacle} = \text{false}$), 
  so it executes \texttt{senseGo} transition: \texttt{Ready} $\to$ \texttt{Go}
  \end{itemize}
\item \emph{Controller:} Stutters (no transitions in \texttt{Sense} phase)
\end{itemize}

Moreover, each instance's \texttt{executed} flag is set to \texttt{true}.
Then, the next scheduler transition is from \texttt{Sense} phase to \texttt{Act} phase, 
which updates the scheduler global phase and resets \texttt{executed} to \texttt{false} for all instances.
This is followed by the \texttt{Act} phase self-transition, where again 
all instances execute their local transitions in the \texttt{Act} phase.
For example, we have:
\begin{itemize}
\item \emph{Controller:} The controller $c_1$ is in \texttt{Idle} state. It evaluates its transition guards for 
the local execution:
  \begin{itemize}
  \item \texttt{actRight} guard: $(\exists s \in \{s_L\} : s.\text{location} = \texttt{NoGo}) \land (\forall s \in \{s_{R1}, s_{R2}\} : s.\text{location} = \texttt{Go})$
  \item This evaluates to: $(s_L.\text{location} = \texttt{NoGo}) \land (s_{R1}.\text{location} = \texttt{Go} \land s_{R2}.\text{location} = \texttt{Go})$, 
  which is true in the current state.
  \end{itemize}
\item \emph{Effect:} Controller executes \texttt{actRight}:
  \begin{itemize}
  \item Sets $c_1.\text{direction} := \texttt{Right}$
  \item Moves to $c_1.\text{location} := \texttt{Moving}$
  \item Sets $s_L.\text{processed} := \text{true}$, $s_{R1}.\text{processed} := \text{true}$ and $s_{R2}.\text{processed} := \text{true}$
  \end{itemize}
\item \emph{Sensors:} Check for \texttt{processed} events:
  \begin{itemize}
  \item $s_L$ has $\text{processed} = \text{true}$, executes \texttt{resetNoGo}: \texttt{NoGo} $\to$ \texttt{Ready}
  \item $s_{R1}$ has $\text{processed} = \text{true}$, executes \texttt{resetGo}: \texttt{Go} $\to$ \texttt{Ready}
  \item $s_{R2}$ has $\text{processed} = \text{true}$, executes \texttt{resetGo}: \texttt{Go} $\to$ \texttt{Ready}
  \end{itemize}
\end{itemize}

All sensors are back in \texttt{Ready} state and have their \texttt{processed} event reset to false.
Note that if we had switched the execution order within the \texttt{Act} phase,
the sensors would have stuttered, and the events would have remained true, triggering an additional 
act phase. Instead, we now take the scheduler transition from \texttt{Act} phase to \texttt{Reset} phase,
followed by the \texttt{Reset} phase self-transition; here, only the controller has a transition:
it executes \texttt{reset} transition: \texttt{Moving} $\to$ \texttt{Idle}.

At the end of Cycle 1, the controller has decided to move right and returned to \texttt{Idle} state. 
The direction variable retains the value \texttt{Right}. 
Before starting Cycle 2, the external inputs (sensor obstacle values) may change: 
we might start in a situation where the controller is moving right, and 
right sensors are detecting obstacles, forcing the controller to change direction
at the end of cycle 2.

\subsection{Local Contracts}
We report an example for a valid local contract for the \code{Controller}'s \code{exec} method in the \code{Act} phase. 
Let $\mathit{leftNoGo} \equiv (\exists s \in \code{leftSensors}.\; s.\code{location} = \code{NoGo})$ and 
$\mathit{rightGo} \equiv (\forall s \in \code{rightSensors}.\; s.\code{location} = \code{Go})$. 
A valid postcondition capturing the \code{actRight} case is:
\begin{align*}
T_{\code{Controller}}(\code{Act}) \equiv {} 
  & \big(\code{old}(\code{location}) = \code{Idle} \land 
     \code{old}(\mathit{leftNoGo}) \land \code{old}(\mathit{rightGo})\big) \\
  & \quad \land\big(\code{location} = \code{Moving} \land \code{direction} = \code{Right} \\
  & \qquad \land\; (\forall s \in \code{All}_{\code{Sensor}}.\; s.\code{processed} = 
    (\text{if } s \in \code{allSensors} \text{ then } \code{true} \\
  & \qquad \qquad \text{else } \code{old}(s.\code{processed})))\big) \\
  \lor\; & \ldots \quad \text{(other transitions handled similarly)}
\end{align*}

\section{Automatic Contract Generation}
\label{sec:effect-transformation}


\subsection{Initialization Contract}

The initialization contract $I_C$ for a class $C$ directly reflects
the semantics of the \texttt{init} method.
It is defined as a conjunction of equalities of the form $v = e$, where $v$ is
a field of the class and $e$ is a constant expression.

\subsection{Execute contracts}

Contracts for \texttt{exec} methods are generated in three steps:
(1) symbolic transformation of transition effects,
(2) encoding of individual transitions, and
(3) combination of the latter for the execution contract.

\medskip
\noindent\emph{Step 1: Symbolic transformation of transition effects.}
For each transition effect $\EffectStmt$, we compute a symbolic effect formula
$\text{Effect}_{\EffectStmt}$ that captures the state update performed by executing $\EffectStmt$.
This is done by forward symbolic execution, which computes a symbolic map $M_{\EffectStmt}$
associating each variable with an expression over pre-state values.
Assignments update the corresponding map entry; conditionals yield conditional
expressions; and quantified assignments generate lambda expressions for function
fields.
Method calls occurring in $\EffectStmt$ are inlined by substituting their bodies, so that
the transformation operates on a single expanded statement sequence.

A detailed description of this transformation is provided in
Algorithm~\ref{alg:effect-transform}.
It maintains a \emph{symbolic map} $M$ that tracks the cumulative effect of assignments, handling sequential composition by chaining updates.
Given an effect $E$ of a class $C$, the transformation initializes $M_0(v) = v_{\text{old}}$ for each of the 
variables $v$ that are fields of $C$ and are assigned to in the effect, 
and $M_0(f) = \lambda y. w_{\text{old}}(y)$ for each variable $w$ 
that is updated by a quantified assignment in the effect, 
then computes the final map $M_E = \textsc{Transform}(E, M_0)$.
The \textsc{Subst} function substitutes all variable occurrences using the current symbolic map.

\begin{algorithm}[h]
\caption{Effect Transformation: \textsc{Transform}(Statement $S$, Map $M$)}
\label{alg:effect-transform}
\begin{algorithmic}[1]
\STATE \textbf{Input:} Statement $S$, Symbolic map $M$
\STATE \textbf{Output:} Updated symbolic map $M'$
\STATE
\IF{$S$ is $v := e$}
    \STATE $M'(v) \gets \textsc{Subst}(e, M)$
    \STATE $M'(w) \gets M(w)$ for all $w \neq v$
\ELSIF{$S$ is $v := *$}
    \STATE $M'(v) \gets \varepsilon$ \hfill $\triangleright$ havoc: no constraint on $v$
    \STATE $M'(w) \gets M(w)$ for all $w \neq v$
\ELSIF{$S$ is \textbf{if} $b$ \textbf{then} $S_1$ \textbf{else} $S_2$}
    \STATE $M_1 \gets \textsc{Transform}(S_1, M)$
    \STATE $M_2 \gets \textsc{Transform}(S_2, M)$
    \FOR{each variable $v$}
        \STATE $M'(v) \gets \textbf{if } \textsc{Subst}(b, M) \textbf{ then } M_1(v) \textbf{ else } M_2(v)$
    \ENDFOR
\ELSIF{$S$ is $\forall x \in E_{set} \{ x.f := e \}$ (quantified assignment)}
    \STATE $M'(f) \gets \lambda y. \textbf{if } y \in E_{set} \textbf{ then } \textsc{Subst}(e[x \mapsto y], M) \textbf{ else } M(f)(y)$
    \STATE $M'(w) \gets M(w)$ for all $w \neq f$
\ELSIF{$S$ is $S_1; S_2$}
    \STATE $M_{temp} \gets \textsc{Transform}(S_1, M)$
    \STATE $M' \gets \textsc{Transform}(S_2, M_{temp})$
\ENDIF
\RETURN $M'$
\end{algorithmic}
\end{algorithm}

The resulting effect formula is defined as:
\begin{align*}
	\text{Effect}_{\EffectStmt} \equiv\;&
\bigwedge_{v \in \text{modified(self)}} v = M_{\EffectStmt}(v) \\
&\land
\bigwedge_{w \in \text{modified(other)}}
\forall x :: x \in \texttt{All}.\; x.w = M_{\EffectStmt}(w),
\end{align*}
where $M_{\EffectStmt}(v)$ and $M_{\EffectStmt}(w)$ are the symbolic expressions computed for fields of the
current object and for fields of other objects accessed via set-valued references,
respectively.

\begin{remark}[Strongest Postcondition]
Existentially quantifying the pre-state variables in $\text{Effect}_{\EffectStmt}$ yields a
formula equivalent to the strongest postcondition of the effect $\EffectStmt$.
\end{remark}

\medskip
\noindent\emph{Step 2: Encoding of transitions.}
Each class transition
$t = (l_{\text{start}}, G, l_{\text{end}}, \EffectStmt, p)$
is encoded as a transition formula $\text{Trans}_t$:
\begin{align*}
\text{Trans}_t \equiv\;&
\texttt{old}(\text{location} = l_{\text{start}} \land \widehat{G}_t) \\
&\land \texttt{Effect}_{\EffectStmt} \land \text{location} = l_{\text{end}} \\
&\land \texttt{executed} \\
&\land \texttt{ResetEvents}(G) \land \texttt{Unchanged}.
\end{align*}

The predicate $\widehat{G}_t$ is the \emph{extended guard}, which enforces transition
priority:
\[
\widehat{G}_t \equiv
G \land
\bigwedge_{t' \prec t,\; t'.\text{start} = l_{\text{start}},\; t'.\text{end} = l_{\text{end}}}
\neg G_{t'}.
\]
Here $t' \prec t$ indicates that transition $t'$ is declared before $t$ in the class
definition.
Since transitions are evaluated in declaration order, $t$ may fire only if its
own guard holds and all higher-priority guards with the same source and target
locations are false. The remaining components of $\text{Trans}_t$ serve auxiliary purposes:
\begin{itemize}
\item $\texttt{ResetEvents}(G)$ resets to \texttt{false} any event fields referenced in
the guard $G$; 
\item $\texttt{Unchanged}$ is an expression asserting that all variables not modified
by $E$ retain their pre-state values.
\end{itemize}

\medskip
\noindent\emph{Step 3: Execution Contracts.}
For a given phase $p$, the execution contract $\mathit{Exec}_C^p$ for class $C$ is
defined as the disjunction of all transition formulas associated with phase $p$,
augmented with a stuttering case:
\[
\mathit{Exec}_C^p \equiv
\bigvee_{\substack{t \in \text{transitions}(C) \\ t.\text{phase} = p}} \text{Trans}_t
\;\lor\;
\text{Stutter}.
\]
The predicate $\text{Stutter}$ asserts that all state variables remain unchanged and
applies when no transition in phase $p$ is enabled, i.e., when all corresponding
guards are false.

\subsection{Tick Timer Contracts}
The tick-timer contract $K_C$ for a class $C$ is also generated
automatically.
Its effect consists of decrementing all active timers declared in the class by one and updating the
state of any timer that reaches zero (e.g., by marking it as \texttt{inactive}).
Also this expression can be derived trivially from the class definition.

\input{proof_theorem.tex}

%% file: proof_theorem.tex
\section{Proof of Theorem 1}
\label{sec:proof-theorem}

Here we provide a proof of Theorem 1 (Compositional Soundness), 
clarifying the connection to the local contracts and global entailment checks 
in Section~\ref{subsec:compositional-verification}.

\medskip
\noindent\textbf{Theorem 1 (Compositional Soundness).}
Given a configurable SRA system $\mathcal{S}$, a set of configuration constraints $\Gamma$, 
and a global property $\varphi$, if all local contracts are valid 
and all global entailment checks in Section~\ref{subsubsec:global-checks} succeed for an invariant $\mathit{Inv}$, 
then $\varphi$ holds at the end of every 
scheduling cycle in every run of~$\mathcal{S}$, 
for all configurations satisfying~$\Gamma$.
\begin{proof}
Fix an arbitrary configuration $\mathcal{C}$ such that $\mathcal{C} \models \Gamma$.
We first prove that the invariant $\mathit{Inv}$ holds in all reachable global states.
We then use this fact to establish the safety property $\varphi$.

\medskip
\noindent\emph{Invariant preservation.}
We proceed by induction on the length of a run, showing that for every reachable
global state $S$, we have $\mathcal{C}, S \models \mathit{Inv}$.

\smallskip
\noindent\textbf{Base case.}
Let $S_0$ be the initial global state produced by system initialization.
For each instance $o$ in $\mathcal{C}$, the local state of $o$ in $S_0$ is a valid
post-state of the \texttt{init} method.
By validity of the initialization contracts, $I_C$ holds for all instances in $S_0$.

Moreover, by construction,
$S_0(\texttt{phase}) = l_0$ and all instances satisfy
$\lnot\texttt{executed}$.
Hence $S_0$ satisfies the antecedent of the initialization check
(``Initialization establishes $\mathit{Inv}$'').
By the validity of this check, we conclude
$\mathcal{C}, S_0 \models \mathit{Inv}$.

\smallskip
\noindent\textbf{Inductive step.}
Assume that $\mathcal{C}, S \models \mathit{Inv}$ for some reachable global state $S$.
Let $S'$ be any state reachable from $S$ by one scheduler transition.
We distinguish cases according to the form of the scheduler transition.

\paragraph{Case 1: Scheduler self-loop.}
Suppose $S(\texttt{phase}) = p$ and the scheduler takes the self-loop $p \to p$.
This corresponds to an arbitrary interleaving of \texttt{exec}(p) invocations.

We model this execution as a sequence
\[
S = \sigma_0 \to \sigma_1 \to \cdots \to \sigma_n = S',
\]
where each step $\sigma_i \to \sigma_{i+1}$ corresponds to executing
\texttt{exec}(p) on a single instance $o_i$.

We prove by induction on $i$ that for all $0 \leq i \leq n$:
\begin{enumerate}
  \item $\mathcal{C}, \sigma_i \models \mathit{Inv}$, and
  \item for every instance $c$ that has not yet executed in phase $p$ at $\sigma_i$,
        $\mathcal{C}, \sigma_i \models c.G'$.
\end{enumerate}

\emph{Base case ($i=0$).}
By the outer induction hypothesis, $\mathcal{C}, S \models \mathit{Inv}$.
Since the self-loop is taken, $S$ also satisfies the scheduler guard $g_{p \to p}$.
By the self-loop preservation check, condition~(a) (establishment of $G'$), we obtain
$\mathcal{C}, \sigma_0 \models c.G'$ for all instances $c$.
Thus both claims hold for $\sigma_0$.

\emph{Inductive step.}
Assume the claims hold for $\sigma_i$, with $i < n$.
The state $\sigma_{i+1}$ is obtained by executing \texttt{exec}(p) on instance $o_i$
and setting $o_i.\texttt{executed}$ to $\texttt{true}$.

By validity of the local execution contract, \texttt{exec}(p) satisfies $T_C(p)$.
Since the antecedent of the self-loop main preservation entailment holds at $\sigma_i$,
its validity yields
$\mathcal{C}, \sigma_{i+1} \models \mathit{Inv}$.

Moreover, for any instance $c \neq o_i$ that has not yet executed in phase $p$,
the self-loop preservation check, condition~(b) (stability) applies,
with $d$ instantiated as $o_i$.
Thus $\mathcal{C}, \sigma_{i+1} \models c.G'$ holds for all such instances.

This completes the inner induction, and therefore
$\mathcal{C}, S' \models \mathit{Inv}$ holds after the self-loop.
\paragraph{Case 2: Scheduler phase change.}
Suppose the scheduler transitions from phase $p$ to phase $p'$ with guard
$g_{p \to p'}$.
By the outer induction hypothesis, we have
\[
\mathcal{C}, S \models
\Gamma \land \mathit{Inv} \land
\texttt{phase} = p \land g_{p \to p'}.
\]

We distinguish cases depending on the target location.

\medskip
\noindent\emph{Case 2a: Non-final transition ($p' \neq l_f$).}
In this case, $S'$ is obtained by updating the scheduler location to $p'$
and resetting all \texttt{executed} flags to \texttt{false}.
Thus the antecedent of the phase-transition non-final entailment is satisfied.
By its validity, we conclude
\[
\mathcal{C}, S' \models \mathit{Inv}.
\]

\medskip
\noindent\emph{Case 2b: Final transition ($p' = l_f$).}
In this case, the scheduler transitions to the final location and additionally
resets all \texttt{executed} flags to \texttt{false} and invokes \texttt{tick}
on all instances.
By validity of the local tick contracts, $K_C$ holds for all instances.
Thus the antecedent of the phase-transition final entailment holds in $S'$.
By its validity, we conclude
\[
\mathcal{C}, S' \models \mathit{Inv}.
\]

\medskip
\noindent\emph{Case 2c: Reset transition ($l_f \to l_0$).}
Finally, suppose the scheduler transitions implicitly from the final location
$l_f$ back to the initial location $l_0$.
By the outer induction hypothesis, we have
\[
\mathcal{C}, S \models
\Gamma \land \mathit{Inv} \land \texttt{phase} = l_f.
\]

By construction of the reset transition, external input variables may change,
while all non-input fields remain unchanged.
Thus $S'$ satisfies the predicate $\text{InputChange}$ and
$\texttt{phase} = l_0$.
The antecedent of the phase-transition reset entailment therefore holds.
By its validity, we conclude
\[
\mathcal{C}, S' \models \mathit{Inv}.
\]

This completes the inductive proof that $\mathit{Inv}$ holds in all reachable states.

\medskip
\noindent\emph{Establishing the safety property.}
Let $S_i$ be a global state at the end of a scheduling cycle, i.e.,
$S_i(\texttt{phase}) = l_f$.
By invariant preservation, $\mathcal{C}, S_i \models \mathit{Inv}$.
Moreover, $S_i$ is obtained by a final scheduler transition.
By the validity of the safety-implication check
(``$\mathit{Inv}$ implies the safety property''), we conclude
$\mathcal{C}, S_i \models \varphi$.

Since the configuration $\mathcal{C}$ and the run were arbitrary, the claim follows.
\end{proof}

%% file: appendix-experiments.tex
\section{Detailed Experiment Results}
\label{sec:appendix-experiments}

For the Railway Protection System,
Table~\ref{tab:app-pi-local-contracts} shows the full verification statistics 
for local contracts, and Table~\ref{tab:app-pi-props} the statistics 
for safety properties.
In Table~\ref{tab:app-pi-props}, properties 1--21 are the properties 
verified in~\cite{cav2025}, while properties 22--25 are new.

For the Signal Control Logic,
Table~\ref{tab:app-sa-local-contracts} shows the full verification statistics 
for local contracts.

\setlength{\tabcolsep}{6pt}
\begin{table}[t]
    \caption{Verification statistics for local contracts of the RPS}
    \centering
    \resizebox{0.85\textwidth}{!}{ 
    \begin{tabular}{c|c|c|rr|rr|rr|rr}
        \toprule
        Version & Class & \# Ass. & 
        Tot. time & Tot. RC & 
        Exec. time & Exec. RC & 
        Guard time & Guard RC & 
        Eff. time & Eff. RC \\
        \hline
    \multirow{7}{4em}{\centering List} 
     & 1 & 88 & 1380s & 2710.11M & 7s & 22.85M & 1370s & 2678.98M & 1s & 3.57M \\ 
     & 2 & 102 & 1069s & 2154.83M & 6s & 20.47M & 1059s & 2123.46M & 2s & 5.22M \\ 
     & 3 & 155 & 1036s & 2229.82M & 11s & 33.91M & 1020s & 2181.6M & 2s & 6.83M \\ 
     & 4 & 73 & 3s & 7.94M & 1s & 2.0M & 0s & 0.46M & 1s & 2.5M \\ 
     & 5 & 31 & 1s & 1.13M & 0s & 0.13M & 0s & 0.01M & 0s & 0.17M \\ 
     & 6 & 128 & 14s & 48.75M & 3s & 11.71M & 0s & 0.5M & 9s & 31.99M \\ 
     & 7 & 415 & 730s & 2283.95M & 322s & 633.21M & 15s & 49.47M & 381s & 1567.6M \\ 
     & Sum & 994 & 4233s & 9436.56M & 350s & 724.29M & 3465s & 7034.48M & 396s & 1617.89M \\
        \hline
    \multirow{7}{4em}{\centering List+QG} 
     & 1 & 88 & 26s & 78.58M & 7s & 20.68M & 16s & 50.87M & 1s & 3.38M \\ 
     & 2 & 102 & 22s & 72.42M & 6s & 18.95M & 13s & 44.03M & 2s & 5.02M \\ 
     & 3 & 155 & 29s & 84.4M & 10s & 28.54M & 15s & 43.23M & 2s & 6.62M \\ 
     & 4 & 73 & 3s & 6.6M & 0s & 1.65M & 0s & 0.2M & 1s & 2.4M \\ 
     & 5 & 31 & 1s & 1.13M & 0s & 0.13M & 0s & 0.01M & 0s & 0.17M \\ 
     & 6 & 128 & 13s & 48.75M & 4s & 11.71M & 0s & 0.5M & 8s & 31.99M \\ 
     & 7 & 419 & 831s & 2237.63M & 377s & 590.09M & 16s & 49.57M & 427s & 1567.6M \\ 
     & Sum & 998 & 925s & 2529.53M & 403s & 671.75M & 60s & 188.41M & 442s & 1617.18M \\ 
       \hline
    \multirow{7}{4em}{\centering Set} 
     & 1 & 84 & 44s & 117.62M & 39s & 105.98M & 1s & 3.59M & 1s & 3.4M \\ 
     & 2 & 98 & 43s & 113.23M & 39s & 98.74M & 1s & 3.85M & 2s & 5.04M \\ 
     & 3 & 148 & 59s & 144.85M & 52s & 125.79M & 1s & 5.0M & 2s & 6.62M \\ 
     & 4 & 69 & 3s & 7.53M & 1s & 1.91M & 0s & 0.29M & 1s & 2.41M \\ 
     & 5 & 29 & 1s & 1.12M & 0s & 0.14M & 0s & 0.0M & 0s & 0.17M \\ 
     & 6 & 126 & 13s & 47.73M & 3s & 12.06M & 0s & 0.45M & 8s & 30.73M \\ 
     & 7 & 407 & 565s & 1721.96M & 211s & 330.57M & 3s & 8.2M & 342s & 1353.57M \\ 
     & Sum & 963 & 729s & 2154.06M & 345s & 675.18M & 7s & 21.37M & 356s & 1401.94M \\
       \hline
    \multirow{7}{4em}{\centering Set+QG}
     & 1 & 84 & 5s & 13.87M & 2s & 5.64M & 1s & 1.3M & 1s & 3.34M \\ 
     & 2 & 98 & 6s & 17.91M & 2s & 7.11M & 1s & 1.49M & 2s & 4.98M \\ 
     & 3 & 148 & 9s & 26.83M & 4s & 12.38M & 1s & 1.89M & 2s & 6.54M \\ 
     & 4 & 69 & 2s & 6.34M & 0s & 1.49M & 0s & 0.18M & 1s & 2.36M \\ 
     & 5 & 29 & 1s & 1.12M & 0s & 0.14M & 0s & 0.0M & 0s & 0.17M \\ 
     & 6 & 126 & 13s & 47.73M & 3s & 12.06M & 0s & 0.45M & 8s & 30.73M \\ 
     & 7 & 407 & 584s & 1738.14M & 223s & 349.33M & 3s & 8.4M & 348s & 1353.57M \\ 
     & Sum & 963 & 620s & 1851.96M & 234s & 388.14M & 5s & 13.7M & 362s & 1401.69M \\  
        \bottomrule
    \end{tabular}
    }
    \label{tab:app-pi-local-contracts}
\end{table}

\begin{table}[t]
    \caption{Verification statistics for safety properties of the RPS}
    \centering
    \resizebox{0.85\textwidth}{!}{ 
    \begin{tabular}{c|rr|rr|rr|rr|rr|rr|rr|rr}
        \toprule
        \multirow{3}{4em}{\centering Prop.} & 
        \multicolumn{4}{c|}{\textbf{List}} & 
        \multicolumn{4}{c|}{\textbf{List+QG}} & 
        \multicolumn{4}{c|}{\textbf{Set}} & 
        \multicolumn{4}{c}{\textbf{Set+QG}} \\
        &
        \multicolumn{2}{c}{Trans.} & \multicolumn{2}{c|}{Exec.} &
        \multicolumn{2}{c}{Trans.} & \multicolumn{2}{c|}{Exec.} &
        \multicolumn{2}{c}{Trans.} & \multicolumn{2}{c|}{Exec.} &
        \multicolumn{2}{c}{Trans.} & \multicolumn{2}{c}{Exec.} \\
        & 
        \multicolumn{1}{c}{Time} & \multicolumn{1}{c|}{RC} & 
        \multicolumn{1}{c}{Time} & \multicolumn{1}{c|}{RC} & 
        \multicolumn{1}{c}{Time} & \multicolumn{1}{c|}{RC} & 
        \multicolumn{1}{c}{Time} & \multicolumn{1}{c|}{RC} &
        \multicolumn{1}{c}{Time} & \multicolumn{1}{c|}{RC} & 
        \multicolumn{1}{c}{Time} & \multicolumn{1}{c|}{RC} & 
        \multicolumn{1}{c}{Time} & \multicolumn{1}{c|}{RC} & 
        \multicolumn{1}{c}{Time} & \multicolumn{1}{c}{RC} \\
        \hline
        1  &  475s  &  1458M  &  15s  &  44M  &  356s  &  856M  &  18s  &  40M  &  388s  &  991M  &  24s  &  59M  &  195s  &  480M  &  18s  &  40M  \\
        2  &  576s  &  1399M  &  18s  &  47M  &  333s  &  792M  &  17s  &  38M  &  379s  &  959M  &  24s  &  58M  &  215s  &  555M  &  17s  &  36M  \\
        3  &  1641s  &  4394M  &  80s  &  200M  &  161s  &  383M  &  61s  &  143M  &  1412s  &  4078M  &  91s  &  215M  &  144s  &  381M  &  59s  &  126M  \\
        4  &  579s  &  1368M  &  20s  &  58M  &  370s  &  863M  &  14s  &  28M  &  357s  &  888M  &  22s  &  52M  &  237s  &  528M  &  15s  &  32M  \\
        5  &  564s  &  1349M  &  19s  &  56M  &  351s  &  788M  &  13s  &  26M  &  381s  &  881M  &  20s  &  48M  &  205s  &  473M  &  15s  &  33M  \\
        6  &  596s  &  1412M  &  17s  &  43M  &  361s  &  827M  &  15s  &  29M  &  402s  &  931M  &  23s  &  56M  &  232s  &  511M  &  15s  &  31M  \\
        7  &  568s  &  1346M  &  21s  &  61M  &  347s  &  854M  &  11s  &  27M  &  391s  &  866M  &  23s  &  59M  &  238s  &  539M  &  18s  &  46M  \\
        8  &  567s  &  1348M  &  17s  &  45M  &  341s  &  851M  &  13s  &  32M  &  428s  &  1018M  &  21s  &  53M  &  213s  &  480M  &  16s  &  35M  \\
        9  &  562s  &  1315M  &  18s  &  50M  &  293s  &  728M  &  12s  &  29M  &  382s  &  886M  &  20s  &  48M  &  236s  &  524M  &  20s  &  49M  \\
        10  &  425s  &  1030M  &  20s  &  48M  &  312s  &  794M  &  11s  &  33M  &  308s  &  755M  &  21s  &  52M  &  164s  &  398M  &  20s  &  42M  \\
        11  &  403s  &  1004M  &  18s  &  50M  &  250s  &  650M  &  12s  &  29M  &  288s  &  722M  &  20s  &  49M  &  136s  &  340M  &  16s  &  34M  \\
        12  &  398s  &  987M  &  17s  &  45M  &  308s  &  796M  &  13s  &  33M  &  300s  &  750M  &  21s  &  49M  &  133s  &  323M  &  16s  &  35M  \\
        13  &  606s  &  1498M  &  18s  &  45M  &  303s  &  770M  &  16s  &  43M  &  332s  &  856M  &  24s  &  59M  &  184s  &  448M  &  20s  &  43M  \\
        14  &  550s  &  1260M  &  19s  &  49M  &  321s  &  819M  &  15s  &  39M  &  320s  &  798M  &  25s  &  60M  &  210s  &  523M  &  17s  &  37M  \\
        15  &  2224s  &  5580M  &  98s  &  233M  &  161s  &  405M  &  42s  &  106M  &  1592s  &  4490M  &  90s  &  209M  &  135s  &  348M  &  43s  &  96M  \\
        16  &  668s  &  1554M  &  18s  &  46M  &  305s  &  769M  &  16s  &  41M  &  331s  &  834M  &  22s  &  51M  &  177s  &  448M  &  16s  &  43M  \\
        17  &  563s  &  1276M  &  19s  &  48M  &  327s  &  818M  &  15s  &  38M  &  292s  &  774M  &  32s  &  75M  &  208s  &  523M  &  14s  &  37M  \\
        18  &  1750s  &  4534M  &  116s  &  277M  &  148s  &  402M  &  72s  &  188M  &  1703s  &  4460M  &  92s  &  226M  &  133s  &  343M  &  43s  &  115M  \\
        19  &  580s  &  1362M  &  18s  &  45M  &  270s  &  743M  &  17s  &  39M  &  305s  &  813M  &  24s  &  60M  &  208s  &  532M  &  13s  &  37M  \\
        20  &  502s  &  1406M  &  20s  &  44M  &  265s  &  723M  &  17s  &  37M  &  314s  &  816M  &  23s  &  57M  &  188s  &  472M  &  13s  &  38M  \\
        21  &  1472s  &  4386M  &  101s  &  249M  &  143s  &  379M  &  54s  &  117M  &  1851s  &  4866M  &  84s  &  196M  &  147s  &  383M  &  50s  &  134M  \\
        22  &  108s  &  319M  &  13s  &  29M  &  43s  &  124M  &  9s  &  16M  &  110s  &  334M  &  13s  &  29M  &  39s  &  127M  &  6s  &  16M  \\
        23  &  110s  &  319M  &  11s  &  29M  &  43s  &  124M  &  9s  &  16M  &  108s  &  334M  &  13s  &  30M  &  40s  &  127M  &  6s  &  18M  \\
        24  &  108s  &  319M  &  11s  &  29M  &  42s  &  124M  &  10s  &  16M  &  109s  &  334M  &  14s  &  31M  &  39s  &  127M  &  6s  &  16M  \\
        25  &  111s  &  319M  &  12s  &  29M  &  42s  &  124M  &  9s  &  16M  &  107s  &  334M  &  12s  &  28M  &  41s  &  127M  &  7s  &  18M  \\
     Sum  &  16708s  &  42544M  &  755s  &  1899M  &  6195s  &  15508M  &  514s  &  1195M  &  12888s  &  33769M  &  797s  &  1908M  &  4096s  &  10064M  &  497s  &  1187M  \\
     Max  &  2224s  &  5580M  &  116s  &  277M  &  370s  &  863M  &  72s  &  188M  &  1851s  &  4866M  &  92s  &  226M  &  238s  &  555M  &  59s  &  134M  \\

        \bottomrule
    \end{tabular}
    }
    \label{tab:app-pi-props}
\end{table}

\begin{table}[t]
    \caption{Verification statistics summary for local contracts of the SCL}
    \centering
    \resizebox{0.85\textwidth}{!}{ 
    \begin{tabular}{c|c|c|rr|rr|rr|rr}
        \toprule
        Version & Class & \# Ass. & 
        Tot. time & Tot. RC & 
        Exec. time & Exec. RC & 
        Guard time & Guard RC & 
        Eff. time & Eff. RC \\
        \hline
    \multirow{4}{4em}{\centering List}
     & 1 & 32 & 7s & 24.51M & 0s & 1.11M & 0s & 0.03M & 6s & 22.02M \\ 
     & 2 & 196 & 305s & 1030.73M & 22s & 73.94M & 1s & 1.32M & 279s & 946.44M \\ 
     & 3 & 40 & 7s & 28.08M & 0s & 0.82M & 0s & 0.25M & 6s & 25.46M \\ 
     & 4 & 135 & 682s & 2730.86M & 74s & 270.51M & 0s & 0.8M & 605s & 2452.96M \\ 
     & Sum & 408 & 1001s & 3814.77M & 97s & 346.38M & 1s & 2.4M & 895s & 3446.88M \\ 
        \hline
    \multirow{4}{4em}{\centering List+QG}
     & 1 & 32 & 7s & 24.96M & 0s & 1.39M & 0s & 0.03M & 6s & 22.25M \\ 
     & 2 & 198 & 305s & 1022.16M & 24s & 67.43M & 1s & 0.92M & 277s & 944.87M \\ 
     & 3 & 40 & 7s & 27.75M & 0s & 0.77M & 0s & 0.1M & 6s & 25.46M \\ 
     & 4 & 136 & 605s & 2453.57M & 6s & 20.78M & 0s & 0.56M & 597s & 2426.29M \\ 
     & Sum & 411 & 924s & 3528.9M & 31s & 90.37M & 1s & 1.6M & 886s & 3418.88M \\ 
        \hline
    \multirow{4}{4em}{\centering Set}
     & 1 & 30 & 6s & 23.67M & 0s & 0.92M & 0s & 0.02M & 5s & 21.41M \\ 
     & 2 & 193 & 298s & 1030.77M & 19s & 75.98M & 1s & 1.04M & 275s & 944.86M \\ 
     & 3 & 38 & 7s & 27.95M & 0s & 0.81M & 0s & 0.15M & 6s & 25.46M \\ 
     & 4 & 133 & 658s & 2680.7M & 62s & 239.7M & 0s & 0.7M & 593s & 2434.0M \\ 
     & Sum & 399 & 969s & 3763.57M & 82s & 317.41M & 1s & 1.91M & 880s & 3425.74M \\ 
        \hline
    \multirow{4}{4em}{\centering Set+QG}
     & 1 & 30 & 6s & 23.65M & 0s & 0.97M & 0s & 0.02M & 5s & 21.41M \\ 
     & 2 & 193 & 295s & 1016.75M & 17s & 62.5M & 1s & 0.87M & 275s & 944.59M \\ 
     & 3 & 38 & 7s & 27.73M & 0s & 0.77M & 0s & 0.08M & 6s & 25.46M \\ 
     & 4 & 133 & 601s & 2450.14M & 5s & 16.2M & 0s & 0.53M & 594s & 2427.32M \\ 
     & Sum & 399 & 910s & 3518.73M & 22s & 80.44M & 1s & 1.5M & 880s & 3418.78M \\ 
        \bottomrule
    \end{tabular}
    }
    \label{tab:app-sa-local-contracts}
\end{table}